\documentclass[preprint,12pt]{elsarticle}

\usepackage{amssymb}
\usepackage{amsmath}
\usepackage{amsthm}

\usepackage{xcolor}
\usepackage{tikz}
\newtheorem{proposition}{Proposition}

\newtheorem{remark}{Remark}

\usepackage[font=footnotesize,labelfont=bf]{subcaption}

\journal{Journal of Computational Physics}
\setkeys{Gin}{draft=false}
\usepackage{lineno}

\begin{document}

\begin{frontmatter}

\title{A Lagrangian method for solving the spherical\\shallow water equations using power diagrams}

\author[midd]{Philip Caplan}
\author[midd]{Otis Milliken}
\author[midd]{Toby Pouler}
\author[midd]{Zeyi Tong}
\author[midd]{Col McDermott}
\author[midd]{Sam Millay}
\affiliation[midd]{organization={Middlebury College, Department of Computer Science}, addressline={75 Shannon Street}, city={Middlebury}, state={VT}, postalcode={05753}, country={USA}}

\begin{abstract}
Numerical simulations of the air in the atmosphere and water in the oceans are essential for numerical weather prediction.
The state-of-the-art for performing these fluid simulations relies on an Eulerian viewpoint, in which the fluid domain is discretized into a mesh, and the governing equations describe the fluid motion as it passes through each cell of the mesh.
However, it is unclear whether a Lagrangian viewpoint, in which the fluid is discretized by a collection of particles, can outperform Eulerian simulations in global atmospheric simulations.
To date, Lagrangian approaches have shown promise, but tend to produce smoother solutions.
In this work, a new Lagrangian method is developed to simulate the atmosphere in which particles are represented with spherical power cells.
We introduce an efficient algorithm for computing these cells which are then used to discretize the spherical shallow water equations.
Mass conservation is enforced by solving a semi-discrete optimal transport problem and a semi-implicit time stepping procedure is used to advance the solution in time.
We note that, in contrast to previous work, artificial viscosity is not needed to stabilize the simulation.
The performance of the spherical Voronoi diagram calculation is first assessed, which shows that spherical Voronoi diagrams of 100 million sites can be computed in under 2 minutes on a single machine.
The new simulation method is then evaluated on standard benchmark test cases, which shows that momentum and energy conservation of this new method is comparable to the latest Lagrangian approach for simulating the spherical shallow water equations.
\end{abstract}




\end{frontmatter}


\section{Introduction}

Numerical simulations of the fluid in the atmosphere and oceans are fundamental for numerical weather prediction (NWP) and global climate modeling.
These simulations traditionally inherit the so-called \emph{Eulerian} viewpoint in which the fluid is modeled from the perspective of an observer watching the fluid pass through fixed spatial regions in the domain of interest.
Frameworks such as the Model for Predication Across Scales (MPAS) of the Energy Exascale Earth System Model (E3SM) use a finite volume method to simulate the atmosphere~\cite{Skamarock_2012} and oceans~\cite{Ringler_2013}.
Other discretization methods, such as finite difference or finite element methods are also possible.
In fact, the ice sheet component of E3SM (MPAS-Albany Land Ice)~\cite{Hoffman_2018} uses a finite element discretization. 

Much effort has been put into making finite volume methods for geophysical fluids robust and accurate.
These methods have migrated from using structured grids on the sphere to using arbitrary unstructured polygonal grids and use Arakawa C-grid staggering to discretize the governing equations~\cite{Thuburn_2009,Ringler_2010}.
By using Voronoi diagrams, the TRiSK scheme benefits from the orthogonality property of Voronoi facets (with respect to neighboring site directions) since the C-grid staggering scheme stores velocity components that are perpendicular to the cell facets.
These methods have evolved from using centroidal Voronoi tessellations into using power diagrams, a generalization of Voronoi diagrams in which each site is equipped with a weight.
Power diagrams offer the ability to optimize cell metrics which facilitate the robustness of the solver~\cite{Engwirda_2021}.
The main challenge with these Eulerian schemes lies in the treatment of the advection term and the resulting discretizations are sensitive to cell shapes.
Another challenge lies in the fact that the accuracy of the simulation depends on how well the mesh resolves important flow features.
This is difficult to know beforehand when constructing an initial mesh, though mesh adaptation~\cite{Abdi_2024} offers a possible solution to this problem.

Alternatively, the fluid can be modeled from the \emph{Lagrangian} perspective, in which the fluid is represented as a collection of discrete particles.
It is unknown whether Lagrangian fluid simulations for geophysical applications are capable of outperforming Eulerian ones~\cite{Alam_2008, Capecelatro_2018} in terms of computational efficiency, robustness or accuracy.
In a Lagrangian framework, the governing equations for mass, momentum and energy conservation are described in terms of the material derivative, and these schemes are not as sensitive to the properties of a mesh (if one is used at all).

A popular technique for simulating fluids in a Lagrangian context is the Smoothed Particle Hydrodynamics (SPH) approach, which was initially introduced in 1977 to solve problems in astrophysics~\cite{Gingold_1977}.
SPH has also been used to solve the shallow water equations (SWE), which is a common set of equations when prototyping discretization schemes for atmospheric modeling.
However, SWE-SPH development has primarily focused on realism for computer graphics applications~\cite{Lee_2010} with little attention to the accuracy of the simulations.
Capecelatro introduced an SPH approach for solving the SWE on the sphere~\cite{Capecelatro_2018}, and demonstrated the approach on the commonly used Williamson test suite.
Capacelatro's method is promising, but the robustness of the method appears to be dependent on a parameter embedded within the SPH kernel and the results seem smoother than what others have reported with Eulerian methods.
Furthermore, while an SPH-based approach may appear suitable for atmospheric simulations on the entire (boundary-less) sphere, it would suffer the same difficulties SPH-based methods exhibit when applying the boundary conditions needed to simulate the oceans.

Another option is to represent particles with discrete regions that cover the sphere.
This concept was pioneered by Augenbaum~\cite{Augenbaum_1982,Augenbaum_1984,Augenbaum_1985} in 1982.
One advantage of Augenbaum's work (which is also an important feature of Capecelatro's work) is the fact that the three-dimensional shallow water equations are solved instead of solving a two-dimensional form of the equations in a latitude-longitude coordinate system, thereby avoiding the singularities at the poles.

Voronoi diagrams have recently been used in two dimensions to solve the incompressible~\cite{Duque_2023,Kincl_2024} and compressible~\cite{Kincl_2025} Navier-Stokes equations.
These methods were first proposed by Springel in 2010 and often use a correction procedure (within each time step) to return the particle positions back to their centroids.
These methods naturally conserve mass, and the continuity equation is only needed when simulating compressible flows.
The local mass of each cell can also be conserved by solving a semi-discrete optimal transport problem~\cite{Levy_2015,Levy_2018}, which has also been demonstrated for solving the incompressible Euler equations~\cite{GM_2018}, animating fluids in three dimensions~\cite{Levy_2024, deGoes_2015} and also for reconstructions of the early Universe~\cite{Levy_2021}.

The primary goal of this work is to develop and investigate a new Lagrangian technique for simulating the atmosphere, with the intention of eventually extending this framework to simulate the oceans.
The proposed approach follows from the idea that Voronoi cell masses can be controlled by solving a semi-discrete optimal transport problem.
The use of a Voronoi diagram (more generally, a power diagram) enables the calculation of first-order accurate differential operators needed to discretize the spherical shallow water equations.
A semi-implicit forward Euler method is used to integrate the equations in time, and no artificial viscosity is needed to stabilize the solution.
After describing the components of the new scheme, the method is then evaluated on the standard Williamson test suite.
\section{Methods}
The constrained form of the shallow water equations on a sphere with radius $a$, $\mathbb{S}^2 = \{ \vec{r} = (x, y, z) \in \mathbb{R}^3\,\colon\,\lVert \vec{r} \rVert = a\}$, rotating with an angular velocity of $\vec{\Omega}$ are
\begin{align}
    \frac{\mathrm{d}h}{\mathrm{d}t} &= -\nabla_s\cdot\vec{u}, \label{eq:swm} \\
    \frac{\mathrm{d}\vec{u}}{\mathrm{d}t} &= -g\nabla_s(h + h_s) - f\vec{r}\times\vec{u} - \lVert \vec{u} \rVert^2 \vec{r}, \label{eq:swp}
\end{align}
where $h$ is the depth of the fluid above the surface with height $h_s$ (relative to the surface of the sphere).
The velocity of the fluid $\vec{u}$ is the time derivative of the position of a fluid particle $\vec{r}$.
The Coriolis parameter is denoted by $f$ and will be specified for each problem setup.
The last term in equation \eqref{eq:swp} is a result of constraining particles to remain on the sphere which is attributed to Cot\'e~\cite{Cote_1988} and Augenbaum~\cite{Augenbaum_1982}.
The differential operator on the surface of the sphere is represented by $\nabla_s$.

Eq.~\eqref{eq:swm} is derived from the continuity equation whereas Eq.~\eqref{eq:swp} is derived by conserving momentum.
The left-hand sides of these equations contain the material derivative $\frac{\mathrm{d}}{\mathrm{d}t} = \frac{\partial}{\partial t} + \vec{u}\cdot\nabla_s$.

We will discretize the fluid into a set of $n$ particles associated with areas covering $\mathbb{S}^2$.
Each particle will have a position $\vec{p}_i$, velocity $\vec{u}_i$, and height $h_i$. At the onset of a simulation, these quantities will be explicitly prescribed, and will then evolve according to the discretization of Eqs.~\eqref{eq:swm} and \eqref{eq:swp}.
\subsection{Spherical Voronoi and power diagrams}
We choose to represent the region associated with each particle as a Voronoi cell.
The advantage of this approach is that the collection of Voronoi cells covers $\mathbb{S}^2$ and, as we will see, the cell areas can be controlled to conserve mass.

Specifically, the particle positions initially correspond to the $n$ site locations $\{ \vec{p}_i \in \mathbb{S}^2 \}$ that define a Voronoi diagram on the sphere.
A spherical Voronoi diagram is the collection of spherical Voronoi cells $\{\mathbb{V}_i \subseteq \mathbb{S}^2\}$ which partition the sphere $\mathbb{S}^2$ such that each cell $\mathbb{V}_i$ is the set of all points on $\mathbb{S}^2$ closest to site $\vec{p}_i$:
\begin{align}%
    \label{eq:voronoi-cell}%
    \mathbb{V}_i = \{ \vec{r} = (x, y, z)\in\mathbb{S}^2\,\colon\,\lVert\vec{p}_i - \vec{r}\rVert < \lVert \vec{p}_j - \vec{r}\rVert,\ \forall j \neq i\}.
\end{align}
Each spherical Voronoi cell $\mathbb{V}_i$ is a polygon with vertices on the unit sphere.
We will consider a generalization of Voronoi diagrams in which a weight is assigned to each site so as to offset the distance calculation:
\begin{align}%
    \label{eq:power-cell}%
    \mathbb{P}_i = \{ \vec{r} = (x,y,z) \in\mathbb{S}^2\,\colon\,\lVert\vec{p}_i - \vec{r}\rVert^2 - w_i < \lVert \vec{p}_j - \vec{r}\rVert^2 - w_j,\ \forall j \neq i\}.
\end{align}
In this general setting, $\mathbb{P}_i$ is a spherical power cell, which also has vertices on $\mathbb{S}^2$.
Geometrically, the weights can be interpreted as an offset of the distance calculation such that it is measured from a sphere of radius $\sqrt{w_i}$ centered on $\vec{p}_i$.
When all the weights are identical, a Voronoi diagram is obtained and the distances are measured directly from the site locations.
\subsubsection{Calculating power cells via halfspace clipping}
Several approaches for calculating spherical Voronoi diagrams and Delaunay triangulations have been proposed.
For example, \texttt{STRIPACK}~\cite{Renka_1997} constructs the dual Delaunay triangulation incrementally, using edge swaps to satisfy the Delaunay property after every insertion.
Caroli et el.~\cite{Caroli_2010} also use an incremental algorithm to construct the Delaunay triangulation, however, they exploit the star-shaped property of the polygonal cavity with respect to each inserted vertex to create new Delaunay triangles.
Sugihara~\cite{Sugihara_2002} computes the spherical power diagram by projecting the edges of a three-dimensional convex hull onto a unit sphere.
Jacobsen et al.~\cite{Jacobsen_2011} construct the spherical Delaunay triangulation using a combination of stereographic projections and a planar Delaunay triangulation procedure which is parallelized using the Message Passing Interface (MPI).
All of these methods are \emph{indirect} in the sense that they either rely on the dual Delaunay triangulation or a convex hull.
This makes it difficult to parallelize such approaches because the topology of the triangulation needs to be consistent where different threads insert vertices.

In our proposed particle-based approach, we only need geometric information about the Voronoi or power cells.
Thus, we have developed a \emph{direct} approach to calculate the geometry of these cells which lends to an easily parallizable implementation.
Our method resembles the halfspace clipping procedure used by \texttt{geogram}~\cite{Levy_2013} and \texttt{Voro++}~\cite{Rycroft_2009}, which we extend to the case of a sphere.

The central concept of this approach is that Voronoi and power cells can be computed as the intersection of all halfspaces defined between the corresponding site $\vec{p}_i$ (with weight $w_i$) and every other site $\vec{p}_j$ (with weight $w_j$).
Here, we will focus on the case of power cells since a Voronoi cell can be obtained by simply setting all weights to the same value.
Each halfspace $\mathcal{H}_{ij}^+$ is defined as the region on the positive side of the bisecting plane (called a \emph{bisector}) between $\vec{p}_i$ and $\vec{p}_j$.
Thus each power cell is defined by
\begin{align}
    \label{eq:power-cell-halfspaces}%
    \mathbb{P}_i = \bigcap\limits_{j \neq i}\ \mathcal{H}^+_{ij}, \quad\mathcal{H}_{ij}^+ = \{ \vec{r} \in \mathbb{S}^2\,\colon\,(\vec{r} - \vec{b}_{ij})\cdot \vec{n}_{ij} > 0 \},
\end{align}
where $\vec{n}_{ij} = \vec{p}_i - \vec{p}_j$ and $\vec{b}_{ij} = \frac{1}{2}(\vec{p}_i + \vec{p}_j) + \frac{1}{2}(w_i - w_j) \frac{(\vec{p}_j - \vec{p}_i)}{\lVert\vec{p}_j - \vec{p}_i\rVert^2}$.
This definition directly leads to an implementable algorithm since we can calculate the intersection between all halfspaces to calculate $\mathbb{P}_i$.
Furthermore, each $\mathbb{P}_i$ can be calculated independently of other cells, which lends to a highly parallelizable implementation.

However, it would be inefficient to calculate all such intersections.
In practice, halfspaces (and bisectors) can be classified as \emph{contributing} or \emph{non-contributing}, based on the radius of security theorem~\cite{Levy_2013}.
This theorem applies to sites lifted to $\mathbb{R}^4$, so let $w_{\max}\ge\max\{w_i\}$ and denote the set of lifted sites as $\{\vec{q}_i = (x_i, w_i, z_i, \sqrt{w_{\max} - w_i}) \in \mathbb{R}^4\}$ (for the original sites $\{\vec{p}_i = (x_i, y_i, z_i) \in \mathbb{R}^3\}$).
Furthermore let $\vec{r}_\ell \in \mathbb{R}^4$ be a point $\vec{r} \in \mathbb{S}^2$ which has been lifted to $\mathbb{R}^4$ such that the fourth coordinate of $\vec{r}_\ell$ is zero. A halfspace $\mathcal{H}_{ij}^+$ (and its corresponding bisecting plane) does not contribute to the power cell $\mathbb{P}_i$ if
\begin{align}
    \label{eq:radius-of-security}
    4\max\limits_{\vec{r}_\ell \in \mathbb{P}_i} \lVert \vec{p}_i - \vec{r}\rVert^2 + 4(w_{\max} - w_i) < \lVert \vec{q}_i - \vec{q}_j\rVert^2.
\end{align}
The left-hand side of this expression is a result of simplifying $\lVert\vec{q}_i - \vec{r}_\ell \rVert^2$. As soon as this condition is satisfied, clipping can be terminated and the power cell is complete.
To reach this condition, halfspaces can be traversed in order of increasing distance $\lVert\vec{q}_i - \vec{q}_j\rVert$, so the nearest neighbors of $\vec{q}_i$ are needed. 
\begin{remark}
  \label{rem:voronoi-equivalence}
  The sites do not need to be explicitly lifted to $\mathbb{R}^4$ in order to compute the power cells. Lifting the sites has the advantage that standard nearest neighbor algorithms (e.g. based on a kd-tree) can be used.
  Furthermore, the radius of security should be directly applied in the lifted space, though Eq.~\ref{eq:radius-of-security} can be simplified to only contain terms involving the site coordinates (in $\mathbb{R}^3$) and weights.
  Details about the connection between power cells in $\mathbb{R}^3$ and restricted Voronoi cells in $\mathbb{R}^4$ are provided in \ref{app:lifting}.
\end{remark}

Each power cell is initialized to a large spherical square centered on the site as shown in the leftmost diagram of Fig.~\ref{fig:voronoi-clipping}.
This figure also shows the sequence of steps (top-to-bottom, left-to-right) of this clipping procedure applied to site $\vec{p}_i$.
These steps are illustrated in a two-dimensional square which should be interpreted as an overhead view of the clipping procedure applied to the initial spherical square.
Without loss of generality, weights are omitted for simplicity.

Neighboring sites are traversed in order of increasing distance from the lifted coordinates of $\vec{p}_i$.
At every step of the algorithm, the bisector between $\vec{p}_i$ (with weight $w_i$) and $\vec{p}_j$ (with weight $w_j$) is defined using Eq.~\ref{eq:power-cell-halfspaces}, producing the bisecting plane defined by the normal vector $\vec{n}_{ij}$ and point $\vec{b}_{ij}$.
This plane is then intersected with the current Voronoi cell by identifying whether any edges between adjacent vertices are on opposite sides of the bisecting plane. This classification is performed by evaluating the sign of $\vec{n}_{ij}\cdot(\vec{r}_k - \vec{b}_{ij})$.
If the signs differ, there is an intersection and the list of vertices is updated.

This approach requires calculating the coordinates of any individual $\vec{r}_k \in \mathbb{P}_i$.
In our implementation, we do not directly store the vertex coordinates since it would be wasteful to store both the cell coordinates and the equations of the planes which define the bisectors.
We only store four floating-point values that define each plane (i.e. $\vec{n}_{ij}$ and $-\vec{n}_{ij}\cdot\vec{b}_{ij}$) as well as an unsigned 8-bit integer at each vertex.
The latter is an index into the array of bisecting plane equations and corresponds to the adjacent bisector in the counterclockwise direction around the cell vertices.
Only this bisector index is needed because the adjacent bisector in the clockwise direction can be retrieved from the previous vertex in the list.

Each vertex $\vec{r}_k$ is adjacent to two bisectors from the halfspaces $\mathcal{H}_k$ and $\mathcal{H}_{k - 1}$, which can be retrieved using the index mentioned above.
Let these halfspaces be defined by the pairs ($\vec{n}_k$, $\vec{b}_k$) and ($\vec{n}_{k -1}$, $\vec{b}_{k - 1}$).
To compute the coordinates of $\vec{r}_k$, we first compute the line of intersection $\mathcal{L}(t) = \vec{l}_0 + \vec{l}\,t$ with
\begin{align*}
    \vec{l} = \frac{\vec{n}_k\times\vec{n}_{k-1}}{\vec{n}_k\times\vec{n}_{k-1}},\quad \vec{l}_0 = (\vec{b}_k\cdot\vec{n}_k)(\vec{n}_{k-1}\times\vec{l}) - (\vec{n}_{k - 1}\cdot\vec{b}_{k-1})(\vec{n}_{k}\times\vec{l}).
\end{align*}
This line is then intersected with the sphere centered at the origin (with radius $a$) to obtain two possible intersection points $\vec{l}_0 + (-\beta \pm \sqrt{\beta^2 - \gamma})\,\vec{l}$, where $\beta = \vec{l}\cdot\vec{l}_0$ and $\gamma = \lVert\vec{l}_0\rVert^2 - a^2$. The intersection point closest to $\vec{p}_i$ is used for $\vec{r}_k$.

This calculation is performed whenever the coordinates of the cell are needed, for example, when checking if the radius of security has been reached, or when computing geometric properties of the cell such as the area or moment.
Aside from requiring less memory than explicitly storing the cell vertex coordinates, another advantage to only storing (1) the plane definitions and (2) a counterclockwise plane index (for every vertex) is that topological information about the Voronoi diagram can be readily extracted.
This information is necessary for identifying the sites on opposite sides of a bisector, or for extracting the dual Delaunay triangulation, if necessary.
\begin{figure}
    \centering
    \begin{minipage}{0.24\textwidth}
    \includegraphics[width=\textwidth]{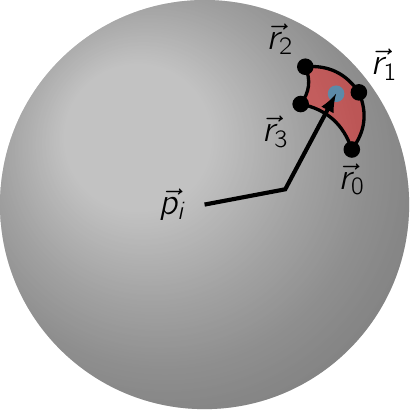}
    \end{minipage}
    \begin{minipage}{0.72\textwidth}
    \includegraphics[width=0.3\textwidth]{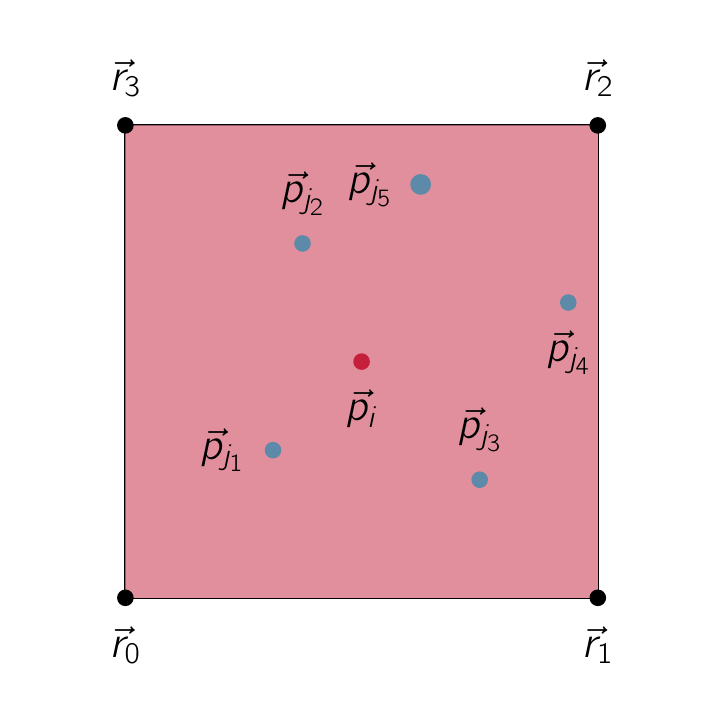}
    \includegraphics[width=0.3\textwidth]{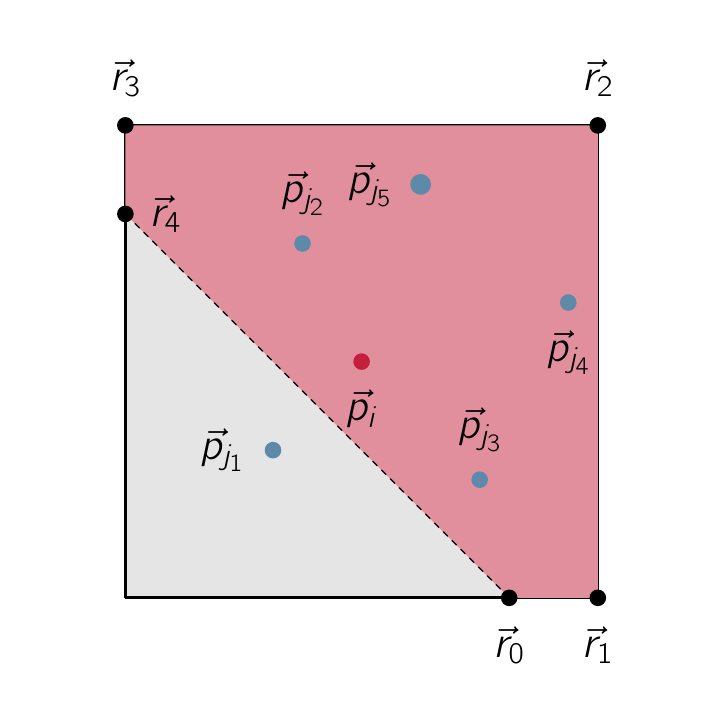}
    \includegraphics[width=0.3\textwidth]{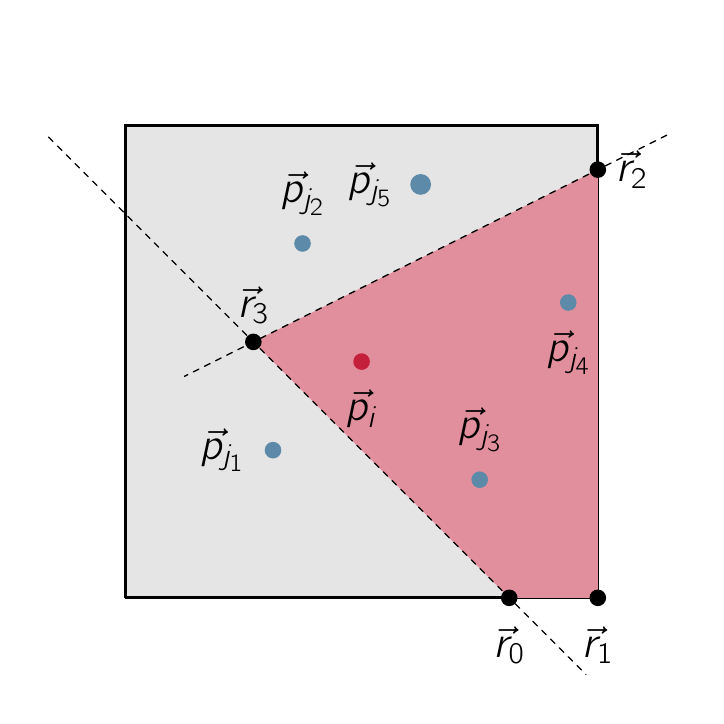} \\
    \includegraphics[width=0.3\textwidth]{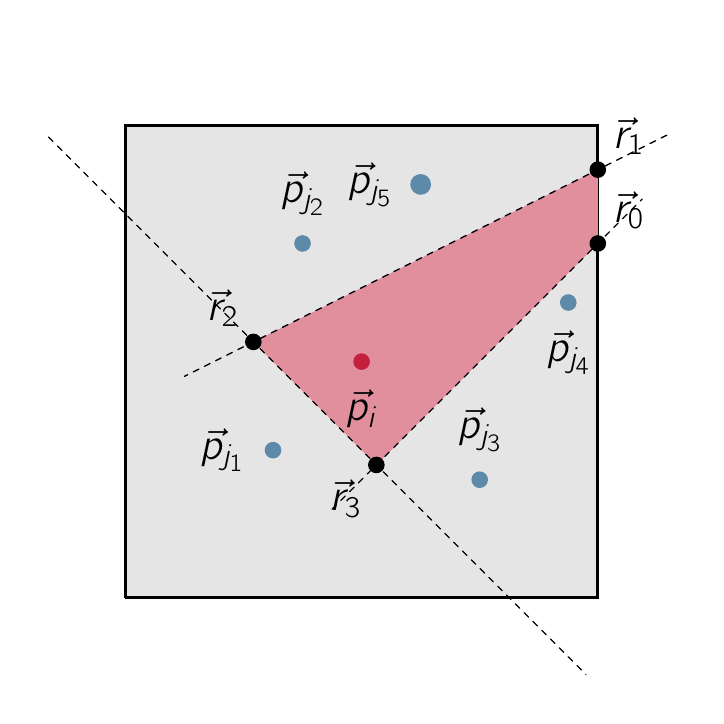}
    \includegraphics[width=0.3\textwidth]{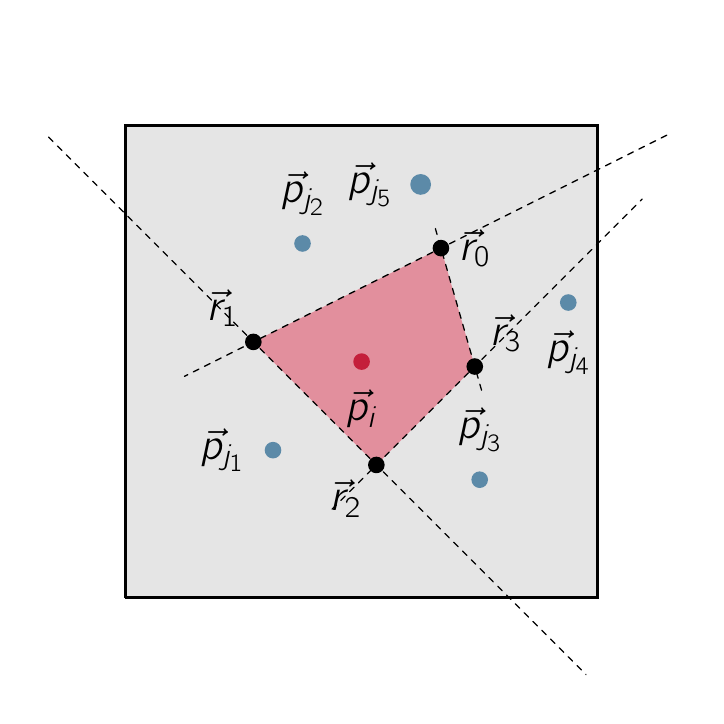}
    \includegraphics[width=0.3\textwidth]{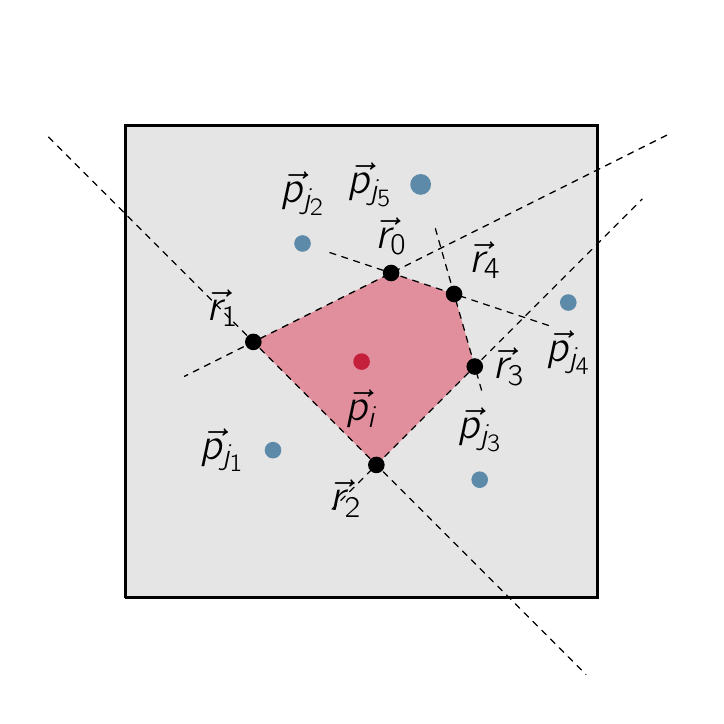}
    \end{minipage}
    \caption{The algorithm for computing an individual power cell (weighted Voronoi cell) starts by initializing the cell to a spherical square centered on the corresponding site.
        From left to right and top-to bottom: the cell is computed by iteratively clipping the current cell with the nearest neighbors of the site until the radius of security is reached.
    }
    \label{fig:voronoi-clipping}
\end{figure}
\subsubsection{Calculating nearest neighbors with a spherical quadtree}
As mentioned in the previous section, the approach used to calculate Voronoi and power diagrams requires calculating the nearest neighboring sites for every site.
A standard approach for calculating these nearest neighbors consists of building a kd-tree and then traversing the resulting kd-tree to retrieve the $k$ neighbors closest to the query site.
However, such an approach is less optimal for the case in which points are distributed close to a sphere~\cite{Caplan_2025}.

As a result, we use a spherical quadtree to accelerate nearest neighbor queries, inspired by the Hierarchical Triangular Mesh~\cite{Kunszt_2003}.
The initial preprocessing step of this approach consists of building a series of successive subdivisions of an octahedron.
\begin{figure}
    \centering
    \includegraphics[width=\textwidth]{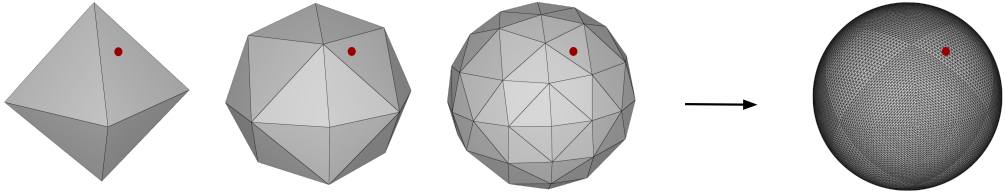}
    \caption{At the onset of a simulation, a spherical quadtree is built by subdividing an octahedron until there are about 20 - 50 sites in each triangle at the finest subdivision.
        When nearest neighbors are needed, this spherical quadtree is then traversed to find the enclosing triangle (in the finest subdivision) of each point.
    }
    \label{fig:sqtree-build}
\end{figure}
During the nearest neighbor query step, each site is located within a triangle of every subdivision by effectively searching through a quadtree defined by the uniform subdivisions of the original 8 triangles (Fig.~\ref{fig:sqtree-build}).
Fig.~\ref{fig:sqtree-locate} shows a site (the red dot) located within the blue triangle at the finest subdivision.
To find the nearest neighbors of this site (in red), all points in triangles sharing a vertex with the blue triangle are accumulated and then sorted.
We have empirically determined that about $20 - 50$ sites per triangle at the finest subdivision level is efficient, so the number of octahedron subdivisions is determined accordingly.
\begin{figure}
    \centering
    \includegraphics[width=\textwidth]{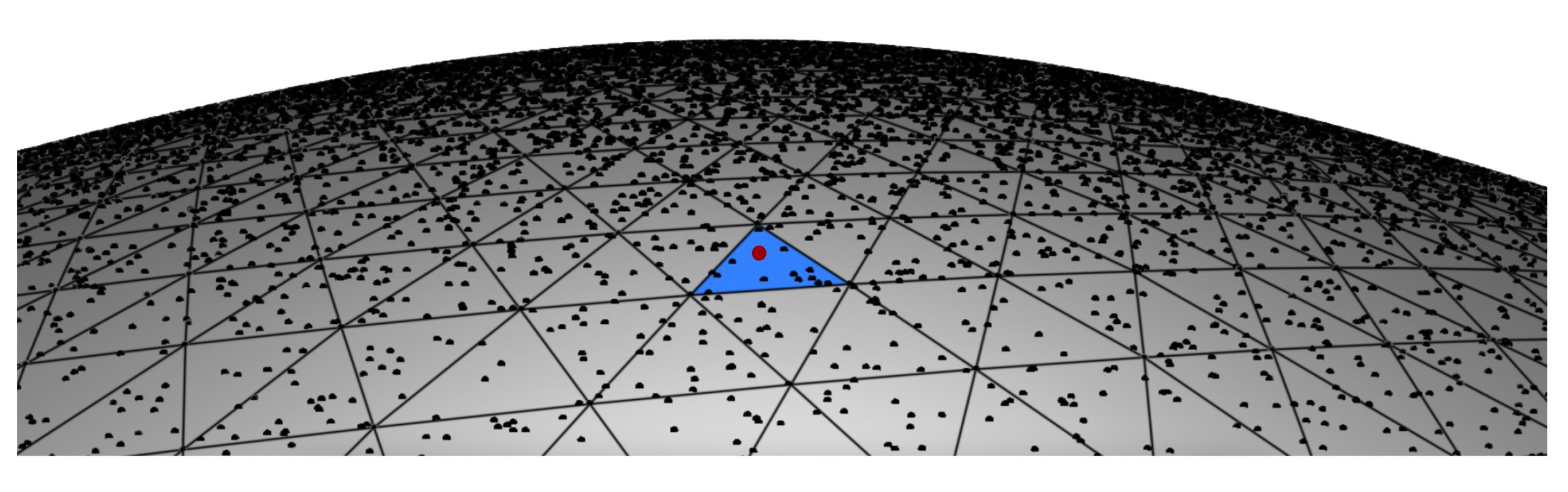}
    \caption{To retrieve the nearest neighbors of a given site (in red), all points in its enclosing triangle (blue), as well as those in triangles sharing a vertex with the blue triangle are accumulated and sorted.}
    \label{fig:sqtree-locate}
\end{figure}
\subsubsection{Performance of the Voronoi diagram calculation}
The performance of the power diagram calculation will now be assessed and compared with existing techniques.
Since other codes only support the calculation of Voronoi diagrams, all weights are set to zero in our analysis.
In general, the weight differences between neighboring sites will be small, so the performance of the power diagram calculation should be similar.

We compare our implementation with \texttt{scipy.spatial.SphericalVoronoi}~\cite{Caroli_2010} and \texttt{MPI-SCVT}~\cite{Jacobsen_mpiscvt}.
All codes were run on a workstation laptop containing an Intel Core Ultra 9 185H chip with 22 threads and 64GB of RAM.
Various point distributions were studied: (1) randomly sampled points from the surface of the sphere, and (2) sites defining a centroidal spherical Voronoi tessellation obtained after 50 iterations of Lloyd relaxation (starting from a random point distribution).
Site coordinates were then re-ordered along the $Z$-curve and, for consistency, the same set of sites was passed to the other algorithms studied here.
For a particular point set, each algorithm was run 10 times and the runtimes were averaged.
The results of the experiments are shown in Fig.~\ref{fig:voronoi-comparison}a for the random point distributions and Fig.~\ref{fig:voronoi-comparison}b for the centroidal point distributions.
The minimum and maximum of the 10 samples are highlighted in the shaded regions and the solid lines depict the average runtimes.
For both types of point distributions, our method is faster than existing algorithms for computing spherical Voronoi diagrams.

Using the same machine described above, we then evaluated the performance of our method as the number of sites grows to over 100 million (M) sites.
In this case, the point distributions are (1) randomly sampled from the surface of the sphere (random) and (2) obtained from the vertices of a uniformly subdivided icosahedron (uniform).
In the latter case, the number of sites is $2 + 10 \cdot 4^s$ where $s$ is the number of icosahedron subdivisions, giving 167,772,162 sites for $s = 12$.
The results of this study are shown in Fig.~\ref{fig:voronoi-comparison}c which compares the time to compute nearest neighbors with the total time to compute the Voronoi diagram (i.e. to compute nearest neighbors \emph{and} to clip the Voronoi cells).
These results highlight that the nearest neighbor calculation is the dominant cost in the Voronoi diagram calculation.
Furthermore, the irregularity in the runtime between 2M and 10M sites may be attributed to the different leaf sizes in the spherical quadtree (i.e. the number of vertices per triangle at the finest subdivision).
The runtime trends are smoother when there are exactly 20 points per triangle at the finest spherical quadtree subdivision, which is the case for the uniform point distributions.
Future work may consist of improving the heuristic used to pre-calculate the number of spherical quadtree subdivisions.
Nonetheless, our algorithm appears to scale linearly with the number of sites, and a spherical Voronoi diagram with 100 million randomly-distributed sites was computed in under 2 minutes. 
\begin{figure}
    \centering
    \begin{subfigure}[t]{0.32\textwidth}
        \begin{tikzpicture}
            \node at (0, 0) {\includegraphics[width=\textwidth]{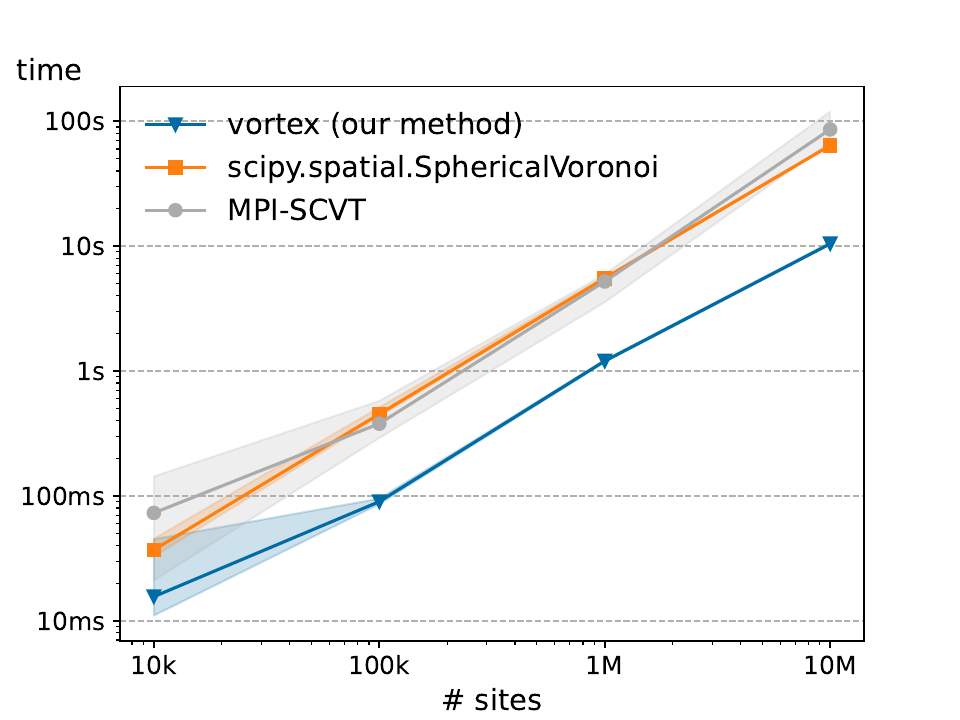}};
            \node at (1.125cm, -0.55cm){\includegraphics[width=0.225\textwidth]{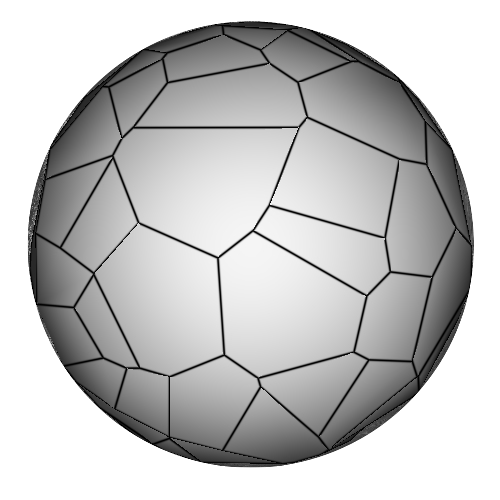}};
        \end{tikzpicture}
        \caption{Random.}
    \end{subfigure}
    \begin{subfigure}[t]{0.32\textwidth}
        \begin{tikzpicture}
            \node at (0, 0) {\includegraphics[width=\textwidth]{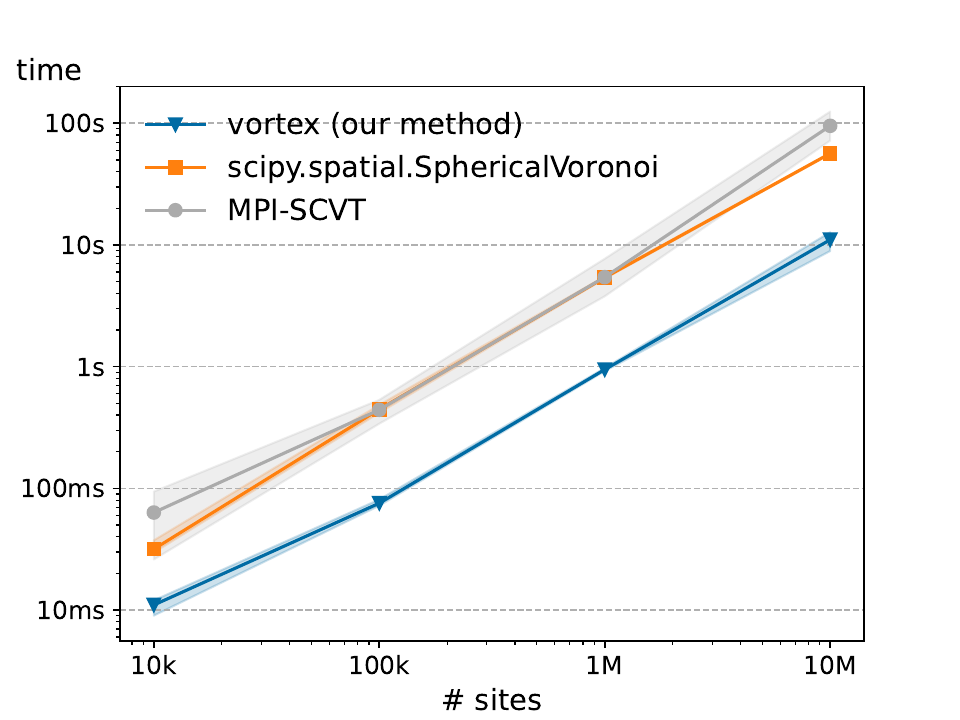}};
            \node at (1.125cm, -0.55cm){\includegraphics[width=0.225\textwidth]{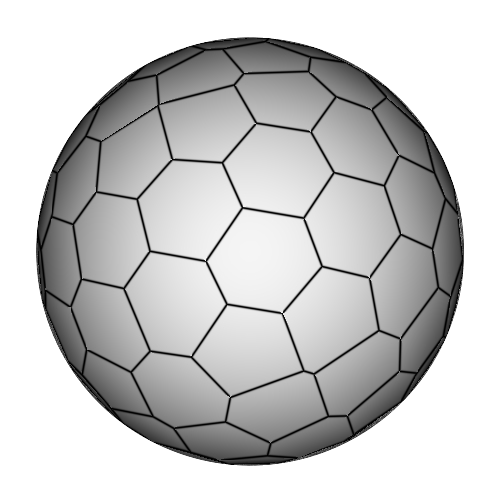}};
        \end{tikzpicture}
        \caption{Centroidal.}
    \end{subfigure}
    \begin{subfigure}[t]{0.32\textwidth}
        \begin{tikzpicture}
            \node at (0, 0) {\includegraphics[width=\textwidth]{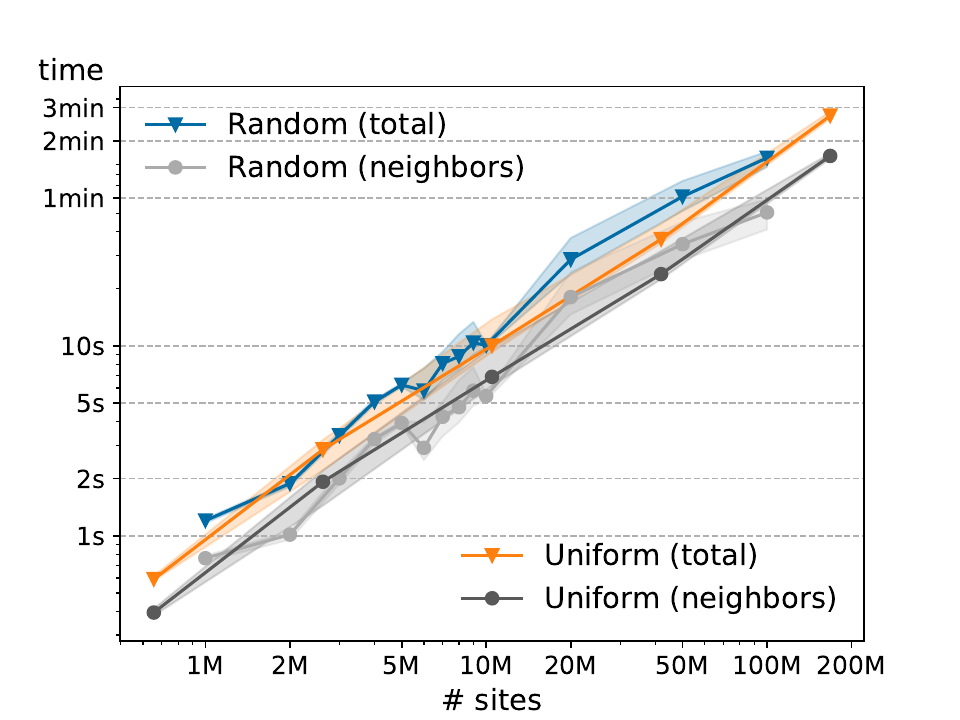}};
            \node at (-1.125cm, 0.25cm){\includegraphics[width=0.15\textwidth]{random.png}};
            \node at (1.125cm, -0.25cm){\includegraphics[width=0.15\textwidth]{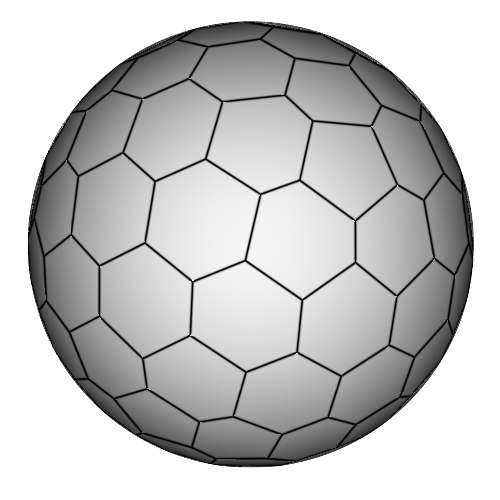}};
        \end{tikzpicture}
        \caption{Performance scaling of our method.}
    \end{subfigure}
    \caption{Comparison of the performance of our method (\texttt{vortex}) with existing techniques for calculating Voronoi diagrams on the sphere.
    Shaded regions highlight the range of runtimes across the 10 samples for each point distribution with the solid lines representing the average runtime.
    Voronoi diagrams embedded in the images are meant to illustrate the differences between each type of point distribution.
    }
    \label{fig:voronoi-comparison}
\end{figure}
\subsection{Discretization of the spherical shallow water equations}
\subsubsection{Discrete gradient and divergence operators on the sphere}
\label{sec:differential-operators}%
Discretizing the spherical shallow water equations requires calculating the gradient and divergence operators on the sphere.
Here, we use the first-order accurate differential operators of Springel~\cite{Springel_2010}, which can be derived using Stokes' theorem.
\begin{figure}
    \centering
    \includegraphics[width=0.5\textwidth]{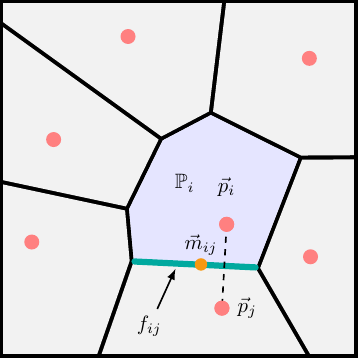}
    \caption{Terms involved in the gradient and divergence calculations.
        The face $f_{ij}$ (in green) is the Voronoi edge between sites $\vec{p}_i$ and $\vec{p}_j$.
        The length of $f_{ij}$ is denoted as $\lvert f_{ij}\rvert$ and its midpoint is $\vec{m}_{ij}$.
        $\mathbb{P}_i$ is the power cell associated with site $\vec{p}_i$ and has an area of $\lvert \mathbb{P}_i\rvert$.
    }
    \label{fig:differential-operators}
\end{figure}
The notation used in the corresponding expressions is described in Fig.~\ref{fig:differential-operators}.
Let $u\,\colon\ \mathbb{S}^2 \to \mathbb{R}$ be a scalar field with $u_i = u(\vec{p}_i)$ the value of $u$ at site $\vec{p}_i$.
Also denote $\vec{p}_{ij} = \vec{p}_j - \vec{p}_i$.
The gradient $\nabla_s u$ at site $\vec{p}_i$ is%
\begin{align}%
    \nabla_s u_i = \frac{\mathbf{P}_i}{\lvert \mathbb{P}_i\rvert}\sum\limits_{f_{ij}\,\in\,\partial\mathbb{P}_i}  \frac{\lvert f_{ij}\rvert}{\lVert\vec{p}_{ij}\rVert}\left((u_j - u_i)\vec{c}_{ij} - \frac{1}{2}(u_i + u_j)\vec{p}_{ij}\right).
\end{align}%
where $\vec{c}_{ij}$ is the vector from the Delaunay edge midpoint to the facet midpoint ($\vec{m}_{ij}$): $\vec{c}_{ij} = \vec{m}_{ij} - \frac{1}{2}(\vec{p}_i + \vec{p}_j)$.
The notation $\lvert\cdot\rvert$ denotes the measure of the corresponding entity, so $\lvert f_{ij}\rvert$ is the length of the facet between sites $\vec{p}_i$ and $\vec{p}_j$ and $\lvert\mathbb{P}_i\rvert$ is the area of the power cell $\mathbb{P}_i$.
The area of the power cell is computed by decomposing $\mathbb{P}_i$ into spherical triangles and adding up the resulting areas. 
A similar expression for the divergence of a vector field $\vec{v}\,\colon \mathbb{S}^2 \to \mathbb{R}^3$ can be obtained with $\vec{v}_i = \vec{v}(\vec{p}_i)$ the value of the vector field at site $\vec{p}_i$:%
\begin{align}%
    \nabla_s \cdot \vec{v}_i = \frac{1}{\lvert \mathbb{P}_i\rvert}\sum\limits_{f_{ij}\,\in\,\partial\mathbb{P}_i}  \frac{\lvert f_{ij}\rvert}{\lVert\vec{p}_{ij}\rVert}\left((\vec{v}_j - \vec{v}_i)\cdot(\mathbf{P}_i\vec{c}_{ij}) - \frac{1}{2}(\vec{v}_i + \vec{v}_j)\cdot(\mathbf{P}_i\vec{p}_{ij})\right),
\end{align}%
Although these expressions were originally derived for Voronoi diagrams~\cite{Springel_2010}, they can also be applied to power diagrams since the main assumption in Springel's derivation is the orthogonality between $\vec{p}_{ij}$ and face $f_{ij}$, which still holds in the case of power cells.
The only difference is that our operators involve a projection to the tangent plane of the unit sphere.
This is performed by the projection matrix $\mathbf{P}_i = \mathbf{I} - \vec{p}_i\vec{p}_i^T$ since $\vec{p}_i$ is the same as the normal vector for the case of $\mathbb{S}^2$.
When calculating the divergence, either vector involved in the dot products can be projected since the projection matrix $\mathbf{P}_i$ is symmetric.

Note that the Voronoi/power diagram does not need to be explicitly stored in our method: only the geometric quantities such as the length and midpoint of each facet $f_{ij}$ and the cell areas need to be stored. Since the Voronoi diagrams are computed on the unit sphere, these operators need to be scaled by $1 / a$ before applying them to problems in physical space.
\subsubsection{Semi-implicit integration in time}
Disturbances in the flow field, such as the conical mountain of Williamson's fifth test case~\cite{Williamson_1992}, can cause gravity waves to develop which travel very quickly.
Numerical schemes for solving the shallow water equations should be robust in the presence of gravity waves, ideally without introducing numerical artifacts that affect the accuracy of the solution.

Initially, we attempted an explicit Euler time integration scheme, and found that artificial viscosity~\cite{Ata_2005} was needed to stabilize the solution.
This had the effect of significantly smoothing out the solution.
Instead, we use a semi-implicit integration scheme to advance the solution in time, and we note that this approach does not require any artificial viscosity for the cases studied in this paper.
Semi-implicit approaches have been previously suggested as a robust method for integrating atmospheric simulations in time~\cite{Weller_2014}.

In what follows, the superscripts $n$ and $n+1$ denote the current (known) and next (unknown) time steps in the simulation, which starts from the known initial condition at time step $0$ of $h^0$ and $\vec{u}^0$.
The time between steps $n$ and $n+1$ is denoted by $\Delta t$. Similar to the method of de Goes~\cite{deGoes_2015}, each particle is advected from the centroid $\vec{c}_{i}$ of the power cell according to its current velocity $\vec{u}_i^n$:
\begin{align}\label{eq:advect-particles}
    \vec{p}_i^{n+1} = \vec{c}_i^n + \Delta t \vec{u}_i^{n}.
\end{align}
To prevent particles from potentially overlapping (which would be problematic when computing the Voronoi diagram), the time step $\Delta t$ is limited such that the velocity can never advect a particle directly onto its closest Voronoi facet.
However, for the time steps considered in the results section, this limiting procedure was never activated.%

In our semi-implicit approach, the mass and momentum equations are discretized as:
\begin{align}
h^{n+1}_i &= h_i^n - h_i^n \Delta t \nabla_s\cdot\vec{u}_i^{n+1}, \label{eq:discretization-h}\\
\vec{u}^{n+1}_i &= \vec{u}_i^n - \Delta t\left[g\nabla_s (h^{n+1}_i + h_s) - f \vec{r}_i^n\times\vec{u}_i^n - \lVert\vec{u}_i^n\rVert^2 \vec{r}_i^n\right].\label{eq:discretization-u}
\end{align}
Since $\nabla_s\cdot\vec{u}_i^{n+1}$ is needed to advance the height field in Eq.~\ref{eq:discretization-h}, taking the divergence of Eq.~\ref{eq:discretization-u} and grouping terms gives:
\begin{align}\label{eq:semi-implicit-div-u}
    \nabla_s\cdot\vec{u}_i^{n+1} = -g \Delta t \nabla_s^2 h^{n+1}_i + \nabla_s\cdot\vec{F}(\vec{r}_i^n, \vec{u}_i^n),
\end{align}
where $\vec{F}(\vec{r}_i^n, \vec{u}_i^n) = \vec{u}_i^n - \Delta t(g\nabla_s h_s + f\vec{r}_i^n\times \vec{u}_i^n + \lVert\vec{u}_i^n\rVert^2 \vec{r}_i^n)$. Substituting Eq.~\ref{eq:semi-implicit-div-u} into Eq.~\ref{eq:discretization-h} yields
\begin{align}\label{eq:semi-implicit-height}
\left(\mathbf{I} - \Delta t^2 gh_i^n \nabla_s^2\right)h_i^{n+1} = h_i^n\left(1 - \Delta t\nabla_s \cdot \vec{F}(\vec{r}_i^n, \vec{u}_i^n)\right).
\end{align}
With $h_i^n$ known, Eq.~\ref{eq:semi-implicit-height} is a linear system of equations for $h^{n+1}$ which we solve using the biconjugate gradient method in \texttt{OpenNL}~\cite{Levy_2006,Levy_OpenNL}. The Laplacian $\nabla_s^2$ is constructed using the operator described in the next section.
\subsubsection{Enforcing mass conservation with optimal transport theory}
Any update to the particle height field may not necessarily conserve the total mass of the system. To improve the conservation properties of our Lagrangian scheme, we propose to conserve the total mass of the particles by solving a semi-discrete optimal transport problem which has been previously applied to simulating fluids~\cite{GM_2018, Levy_2024, deGoes_2015} and in early reconstructions of the Universe~\cite{Levy_2021}.

The main idea is to assign \emph{target} areas ($a_{t,i}$) to each particle in order to conserve the total mass of the system.
Assume each particle starts with a mass of $m_i^0 = \rho a_i^0 h_i^0 = m_i^{n+1}$, with $\rho$ being the constant density of the fluid.
Since every particle has a constant mass throughout the simulation, then the target area is $a^{n+1}_{t,i} = m_i^0 / h_i^{n+1}$.
The optimal transport problem then consists of finding the weights in the power diagram that result in power cells with these target areas.
This can be solved with a Newton-based approach~\cite{Levy_2018, Kitagawa_2016, Milliken_2024}, which quadratically converges to a global minimum by minimizing the following energy functional:
\begin{align}
    \mathcal{E}(\{\mathbb{P}_i\}) = \sum\limits_{i = 1}^n \left[\int\limits_{\mathbb{P}_i}\lVert \vec{p}_i - \vec{x}\rVert^2 - w_i\mathrm{d}\vec{x} + w_i a_{t,i}\right].
\end{align}
At each iteration $k$ of the Newton step, the weights $\mathbf{w} = \{w_i\}$ are updated according to $\mathbf{w}^{k+1} = \mathbf{w}^k + \delta \mathbf{w}$ where $\delta \mathbf{w}$ is obtained by solving:
\begin{align}\label{eq:newton-update}
    \nabla_w^2\mathcal{E}\delta\mathbf{w} = -\nabla_w\mathcal{E}.
\end{align}
The component of the gradient of $\mathcal{E}$ with respect to weight $w_i$ is $\nabla_w \mathcal{E}_i = a_{t,i} - a_{k,i}$, i.e. the difference between the target and current power cell area (at iteration $k$ of the optimization). Using the same notation introduced in Section~\ref{sec:differential-operators}, the components of the Hessian are~\cite{Levy_2018, deGoes_2015}
\begin{align}\label{eq:energy-hessian}
    \nabla_w^2 \mathcal{E}_{ij} = \frac{1}{2}\frac{\lvert f_{ij}\rvert}{\lVert\vec{p_i} - \vec{p}_j\rVert},\quad \nabla_w^2\mathcal{E}_{ii} = -\sum_j\limits\nabla_w^2\mathcal{E}_{ij},
\end{align}
where the sum for the diagonal term is performed over adjacent power cells.
Starting from an initial guess of zero weights, this Newton-based procedure usually converges to $10^{-8}$ in the $L^2$ norm of the gradient $\nabla_w\mathcal{E}$ in a few iterations (about 3 - 5).

It is important to note that the optimal transport problem is further constrained by the fact that all the power cell areas should add up to the total area of the unit sphere. Since we are free to assign the target particle areas in order to conserve mass, this is solved by simply scaling the target areas:
\begin{align}\label{eq:target-areas}
    a_{t,i} = \alpha\frac{m_i^0}{h_{i}^{n+1}}, \quad \alpha = \frac{4\pi}{\sum\limits_i\frac{m_i^0}{h_i^{n+1}}}.
\end{align}
\subsection{Summary}
The proposed Lagrangian approach for simulating the shallow water equations will now be summarized using the equations developed in the previous sections. It is assumed that initial height $h^0$ and velocity data $\vec{u}^0$ are prescribed at some initial particle positions. Here, the particles are initialized as the vertices of a subdivided icosahedron, similar to the particle initialization of Capecelatro~\cite{Capecelatro_2018}. Then, for each time step:
\begin{enumerate}
    \item Advect the particles using Eq.~\ref{eq:advect-particles}.
    \item Update particle heights $h^{n+1}$ using Eq.~\ref{eq:semi-implicit-height}.
    \item Compute the target areas to conserve mass using Eq.~\ref{eq:target-areas}.
    \item While the power cell areas are not sufficiently close to the target areas:
    \begin{enumerate}
        \item Calculate the gradient and Hessian for the Newton step using Eq.~\ref{eq:energy-hessian}.
        \item Solve for $\delta\mathbf{w}$ using Eq.~\ref{eq:newton-update}.
        \item Update the power diagram weights: $\mathbf{w}^{k+1} = \mathbf{w}^k + \delta \mathbf{w}$.
        \item Recalculate the power diagram with the new set of weights.
    \end{enumerate}
    \item Update the particle velocities using Eq.~\ref{eq:discretization-u}.
\end{enumerate}
\section{Spherical Shallow Water Equation benchmarks}
The proposed Lagrangian method will now be applied to some commonly used benchmarks for numerical simulations of the shallow water equations.
The reader is referred to the paper by Williamson et al.~\cite{Williamson_1992} for a complete description of the test cases.
Some of these cases are parametrized by an angle $\alpha$ to test the robustness of the solver in handling the singularities at the poles.
This is not an issue here since all calculations are performed in three-dimensional Cartesian space, so a single value for $\alpha$ will be used instead of performing a parameter sweep.

The method proposed here naturally conserves mass, so only momentum and energy conservation results are reported.
Where appropriate, the solution will either be compared with an analytic one, or will be compared with the solution obtained from an existing solver.
The total momentum ($I_p$) and energy ($I_e$) are:

\begin{align}
    I_p &= \sum\limits_i \lvert \mathbb{P}_i\rvert \left[\lVert\vec{u}_i\rVert + \sqrt{gh_i}\right], \\
    I_e &= \frac{1}{2}\sum\limits_i \lvert\mathbb{P}_i\rvert \left[h_i \lVert\vec{u}_i\rVert^2 + g((h_i + h_{s,i})^2 - h_{s,i}^2)\right].
\end{align}
Using $I_p^0$ and $I_e^0$ to refer to the initial momentum and energy of the particles (computed using the initial condition for each case), conservation will be reported as:
\begin{align}
    \Delta I_p = \frac{I_p - I_p^0}{I_p^0},\quad \Delta I_e = \frac{I_e - I_e^0}{I_e^0}.
\end{align}
For each case, 5, 6, and 7 subdivisions of the icosahedron are used to initialize the particles, resulting in 10,242, 40,962, and 163,842 particles, respectively.
The range of time steps $\Delta t$ studied will also be specified for each case.

The variables $\lambda$ and $\theta$ refer to the longitude and latitude, with $\lambda \in [-\pi,\pi]$ and $\theta \in [\frac{-\pi}{2}, \frac{\pi}{2}]$.
The unit vectors in these directions are denoted by $\vec{e}_\lambda$ and $\vec{e}_\theta$.
The radius of the Earth is $a = 6,371,220\, \mathrm{m}$, $\Omega = 7.292\cdot 10^{-5}, \mathrm{m/s}$ and $g = 9.80616\ \mathrm{m}/\mathrm{s}^2$.
\begin{figure}
    \begin{subfigure}{0.49\textwidth}
        \includegraphics[width=\textwidth]{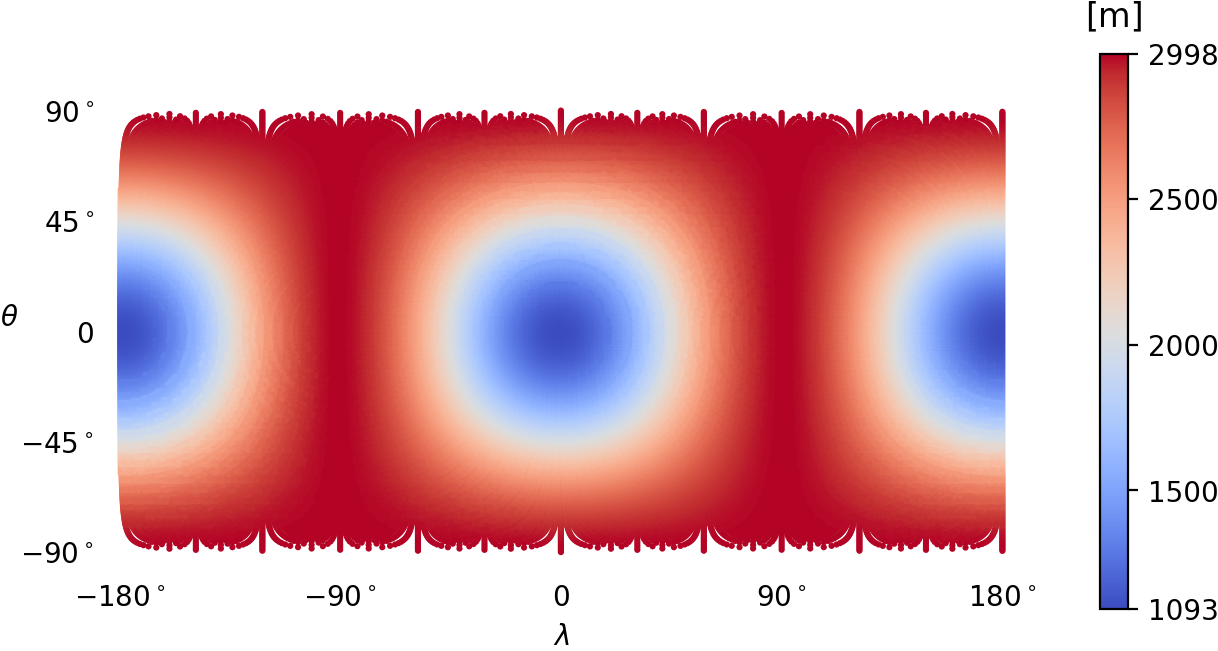}
        \caption{Height at 0 days.}
    \end{subfigure}
    \begin{subfigure}{0.49\textwidth}
        \includegraphics[width=\textwidth]{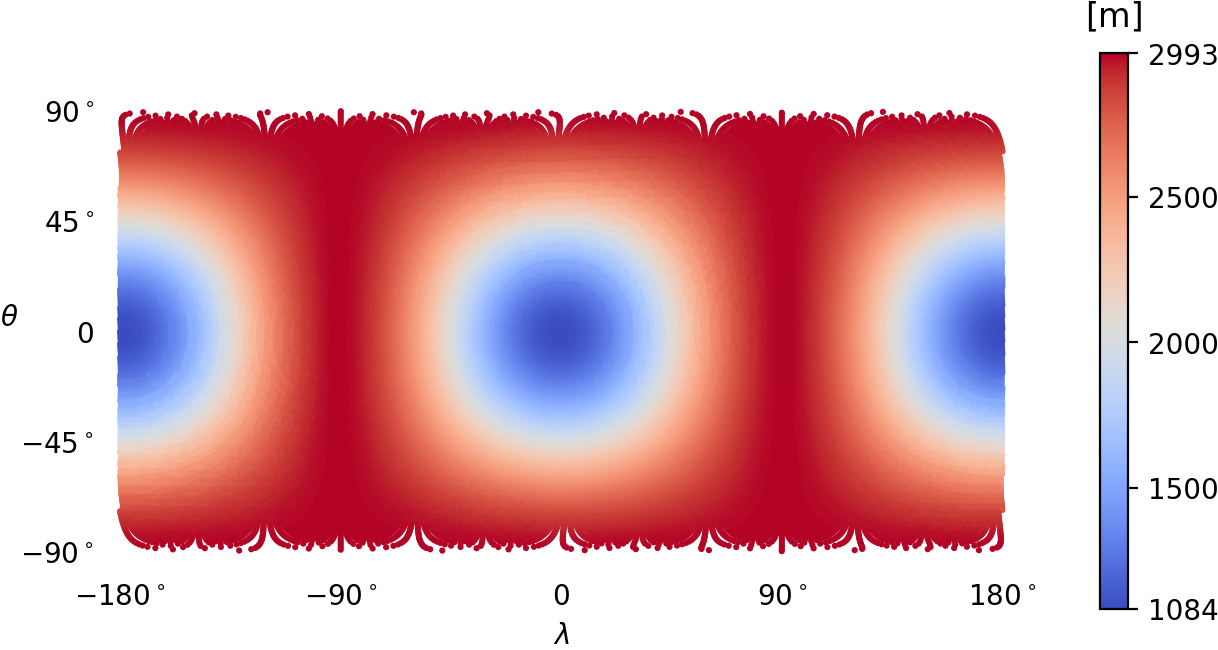}
        \caption{Height at 12 days.}
    \end{subfigure}
    \caption{Height field for Williamson Test Case 2 at the beginning and end of the simulation computed with 163,842 particles and a time step of 30 seconds.}
    \label{fig:w2-height}
\end{figure}
\subsection{Williamson Test Case 2: Steady state zonal geostrophic flow}
The second test case of Williamson et al.~\cite{Williamson_1992} (with $\alpha = \frac{\pi}{2}$) consists of a steady state flow with a velocity of
\begin{align*}
    \vec{u} = u_\lambda \vec{e}_\lambda + u_\theta \vec{e}_\theta,\quad u_\lambda = u_0\sin\theta\cos\lambda,\ u_\theta = -u_0\sin\lambda,
\end{align*}
where $u_0 = 2\pi a / (12\ \mathrm{days})$.
The initial (and analytic) height of a particle is $h = h_0 - \frac{1}{g}(a \Omega u_0 + \frac{1}{2}u_0^2)z$ where $h_0 = 2940 / g$.
The Coriolis parameter is $f = -2\Omega x$.
The analytic height and velocity of a particle consists of evaluating these expressions using the current position of the particle as it moves in time.
Time steps of $\Delta t = 2\,\mathrm{minutes},\ 1\,\mathrm{minute}$, and $30\,\mathrm{seconds}$ were used for 10,242, 40,962, and 163,842 particles, respectively.

Fig.~\ref{fig:w2-height} shows the initial and final (after 12 days) height of each particle which qualitatively shows that these fields are similar.
The evolution of the $\ell^2$ error between the computed and analytic fields is shown Fig.~\ref{fig:w2-error-evolution} and the convergence of the error at day 5 (as recommended by Williamson et al.) is shown in Fig.~\ref{fig:w2-error-convergence}.
The latter shows that the error converges at a rate of about $1$ as the effective size covered by each particle (measured as $(\mathrm{\# particles})^{-1/2})$ decreases.
The conservation of momentum and energy are shown in Fig.~\ref{fig:w2-conservation}, with the 163,842-particle case maintaining a fractional difference of about $10^{-3}$ in the initial energy and momentum of the system.
\begin{figure}
    \centering
    \begin{subfigure}[t]{0.325\textwidth}
        \includegraphics[width=\textwidth]{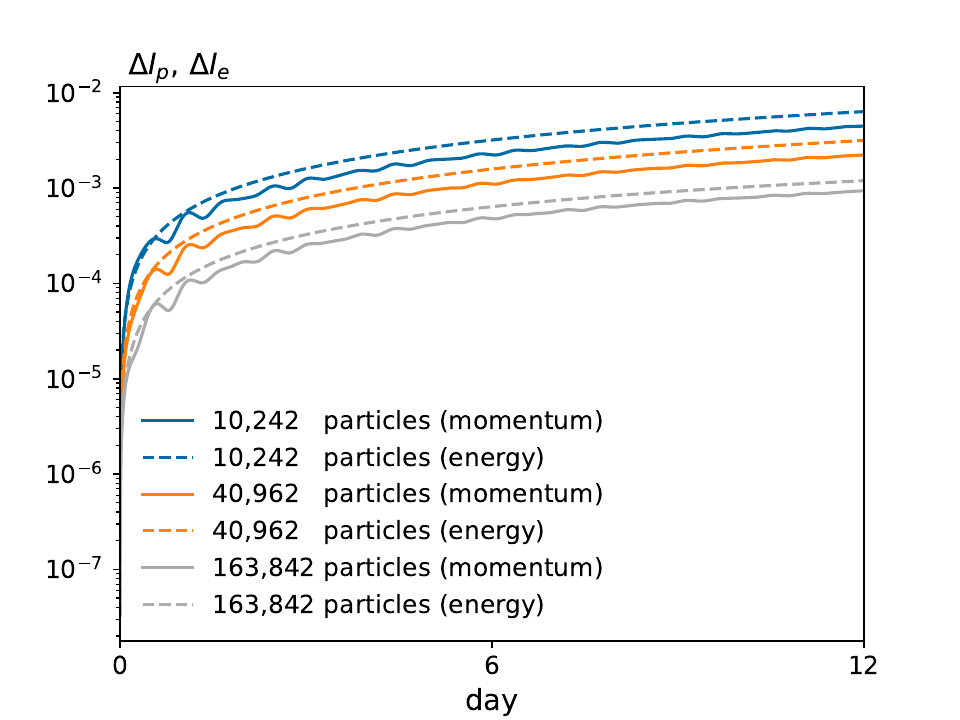}
        \caption{Evolution of momentum and energy conservation.}
        \label{fig:w2-conservation}
    \end{subfigure}
    \begin{subfigure}[t]{0.325\textwidth}
        \includegraphics[width=\textwidth]{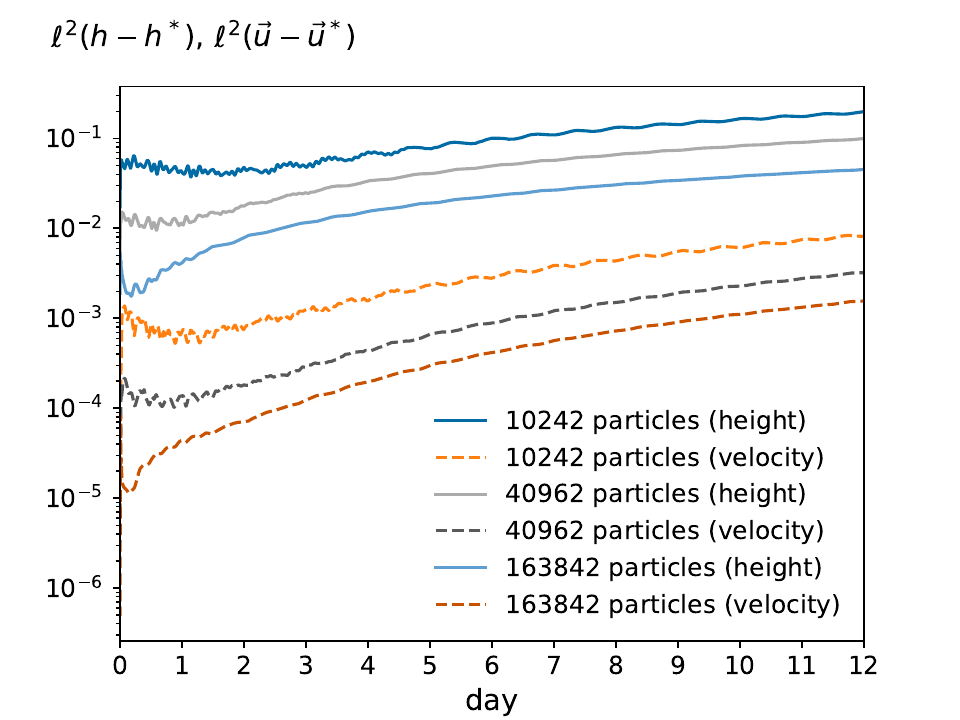}
        \caption{Evolution of the error.}
        \label{fig:w2-error-evolution}
    \end{subfigure}
    \begin{subfigure}[t]{0.325\textwidth}
        \includegraphics[width=\textwidth]{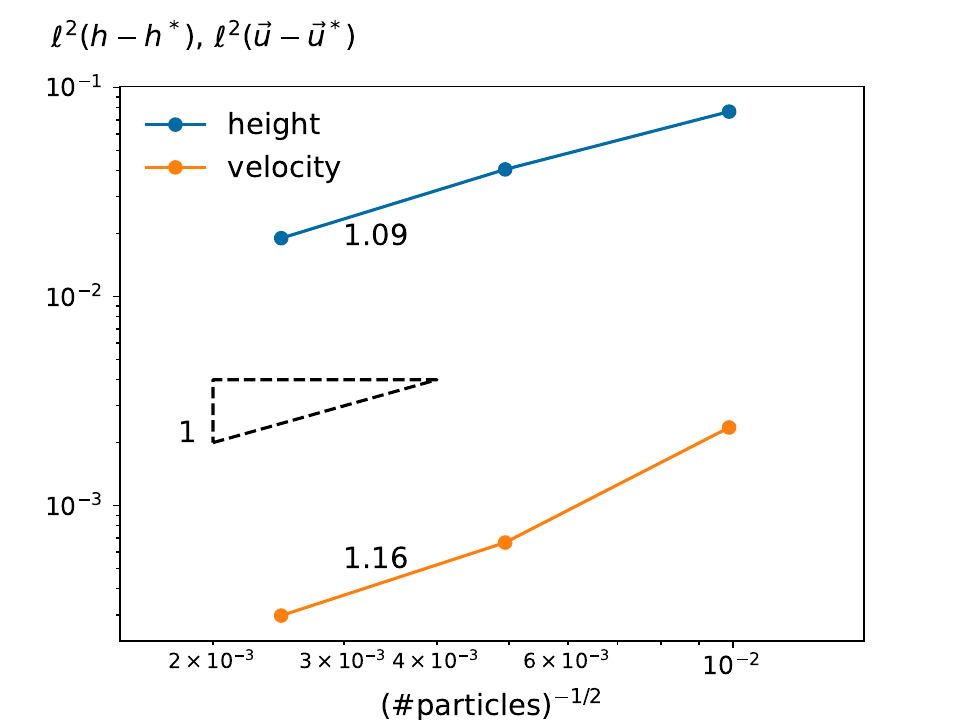}
        \caption{Convergence of the $\ell^2$ error in the solution at 5 days.}
        \label{fig:w2-error-convergence}
    \end{subfigure}
    \caption{Evolution of the conservation of momentum and energy, as well as the $\ell^2$ error in the height and velocity fields for Williamson Test Case 2.
        At day 5, the error in both the height and velocity fields converge with a rate of 1.09 and 1.16, respectively.
    }
    \label{fig:w2-properties}
\end{figure}
\subsection{Williamson Test Case 5: Zonal flow over an isolated mountain}
The next test case introduces a conical mountain centered at $\lambda_0 = -\frac{1}{2}\pi$ and $\theta_0 = \frac{1}{6}\pi$ into a zonal flow field.
Although the usefulness of this case has been debated~\cite{Flyer_2012, Galewsky_2004} because of the non-differentiability of the mountain, it is useful to include since it shows that our proposed method is robust in the presence of such phenomena.
The initial velocity of the particles is $\vec{u} = u_0(-y, x, 0)$ and the Coriolis parameter is $f = 2\Omega z$.
The initial height of the fluid above the surface is $h = h_0 - \frac{1}{g}(\Omega a u_0 + \frac{1}{2}u_0^2)$ with $u_0 = 20\ \mathrm{m/s}$ and $h_0 = 5960\ \mathrm{m}$.
The simulation was run for 15 days with time steps of $\Delta t = 2\,\mathrm{minutes},\ 1\,\mathrm{minute}$, and $30\,\mathrm{seconds}$ for 10,242, 40,962, and 163,842 particles, respectively.
\begin{figure}
    \begin{subfigure}{0.49\textwidth}
        \includegraphics[width=\textwidth]{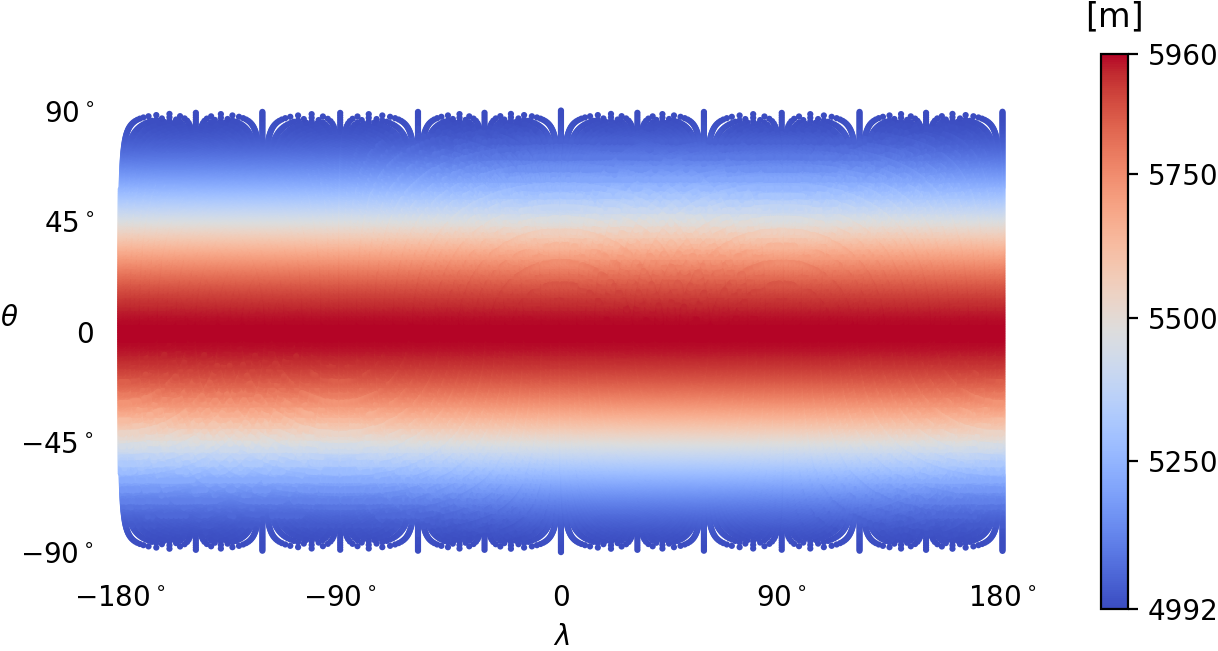}
        \caption{Height at 0 days.}
    \end{subfigure}
    \begin{subfigure}{0.49\textwidth}
        \includegraphics[width=\textwidth]{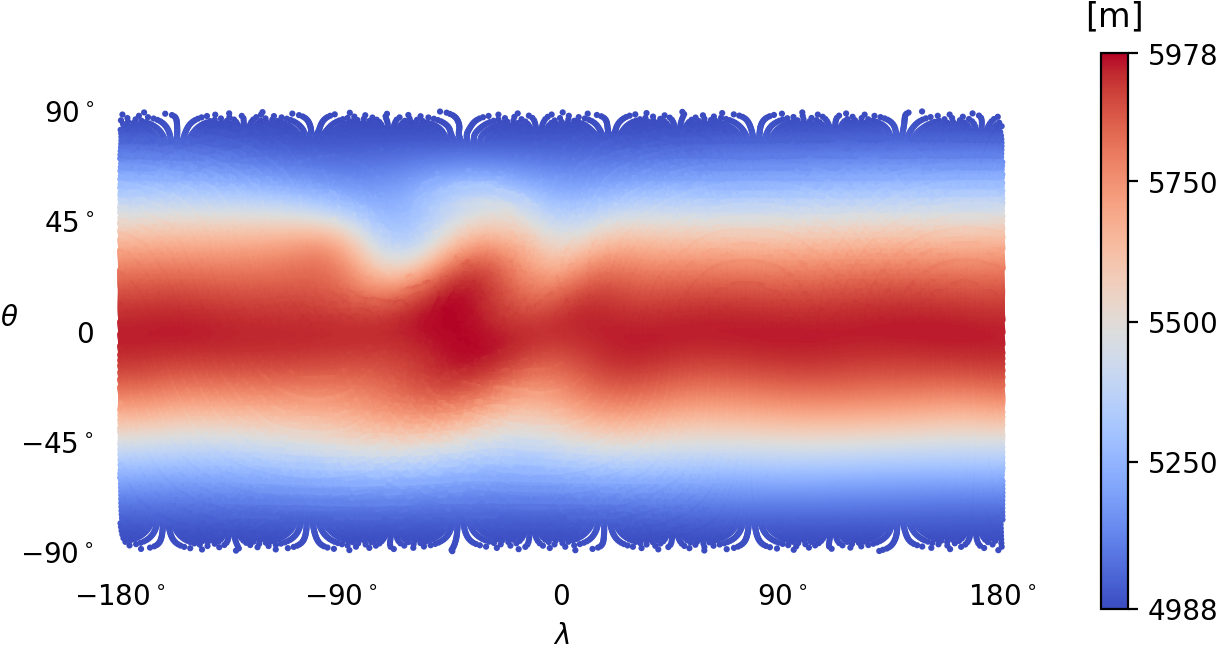}
        \caption{Height at 5 days.}
    \end{subfigure}
    \begin{subfigure}{0.49\textwidth}
        \includegraphics[width=\textwidth]{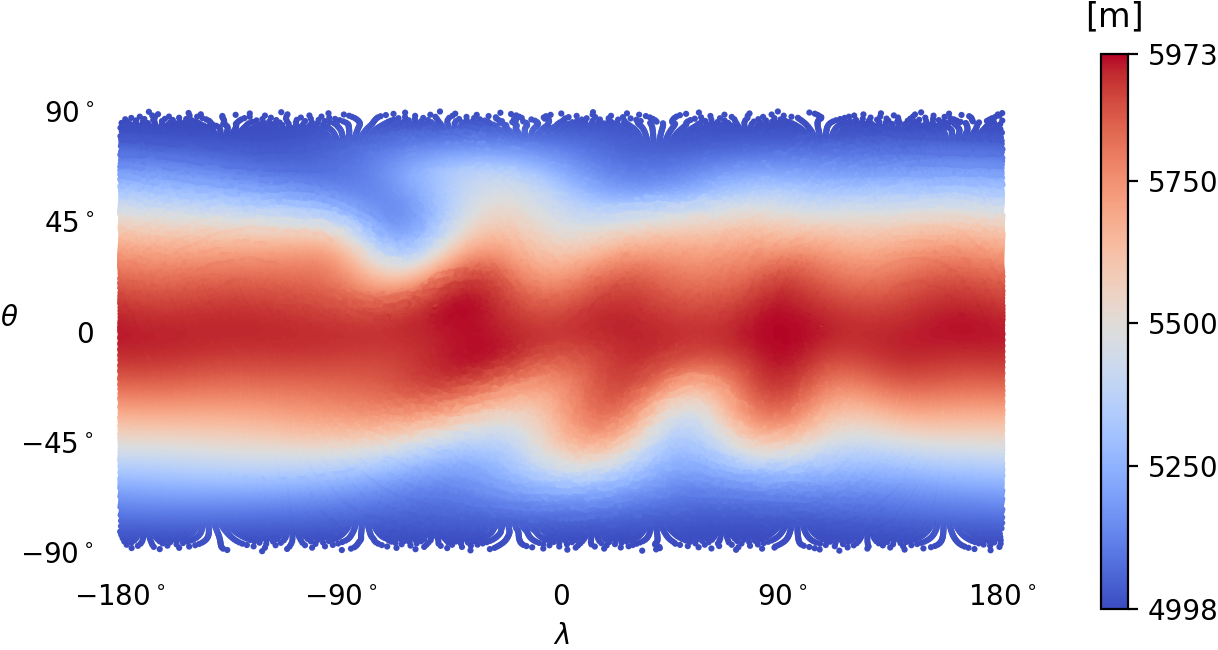}
        \caption{Height at 10 days.}
    \end{subfigure}
    \begin{subfigure}{0.49\textwidth}
        \includegraphics[width=\textwidth]{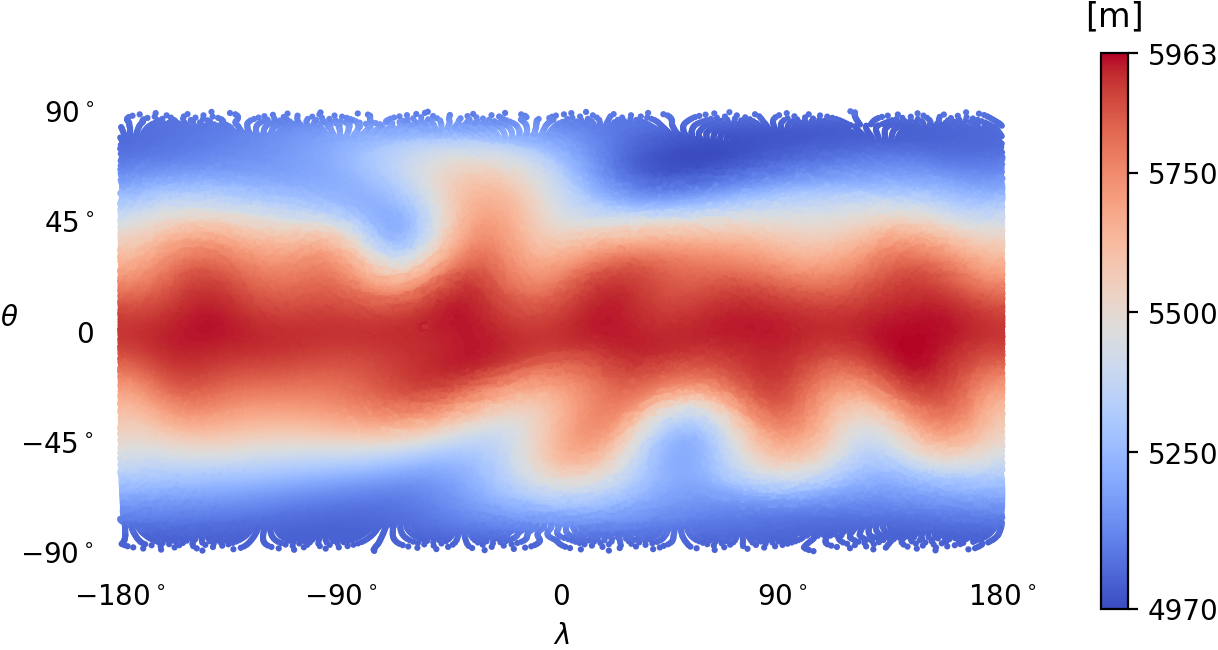}
        \caption{Height at 15 days.}
    \end{subfigure}
    \caption{Height field for Williamson Test Case 5 at 0, 5, 10 and 15 days computed with 163,842 particles and a time step of 30 seconds.}
    \label{fig:w5-h-latlon}
\end{figure}
The resulting solution in latitude-longitude coordinates with 163,842 particles is shown in Fig.~\ref{fig:w5-h-latlon} which is qualitatively consistent with results from other solvers~\cite{Flyer_2012}.
Power diagrams at each day, with cells colored by the corresponding particle height and weight in Cartesian coordinates, are shown in Fig.~\ref{fig:w5-hw-3d}.
At the onset of the simulation, the weights are all zero (blue) but the presence of the mountain causes upstream weights to be higher relative to the other cells.
\begin{figure}
    \begin{subfigure}[t]{0.225\textwidth}
        \includegraphics[width=\textwidth]{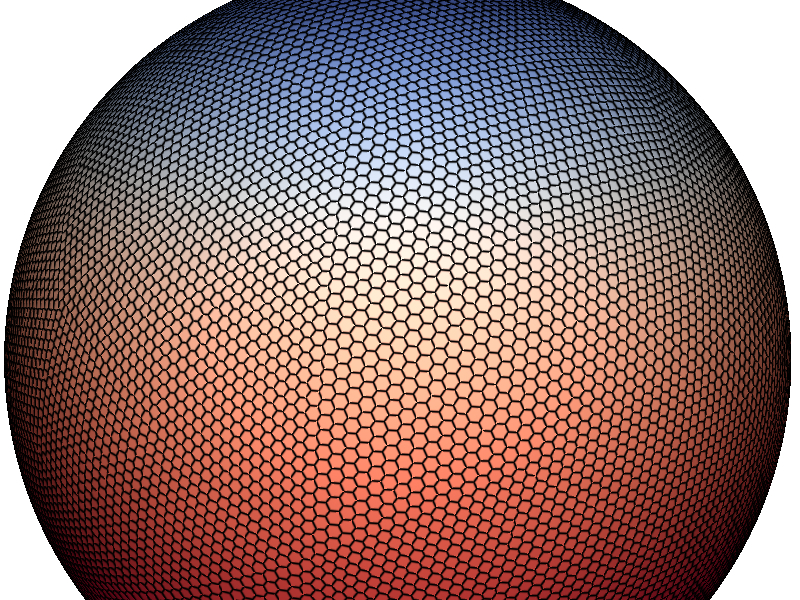}\\
        \includegraphics[width=\textwidth]{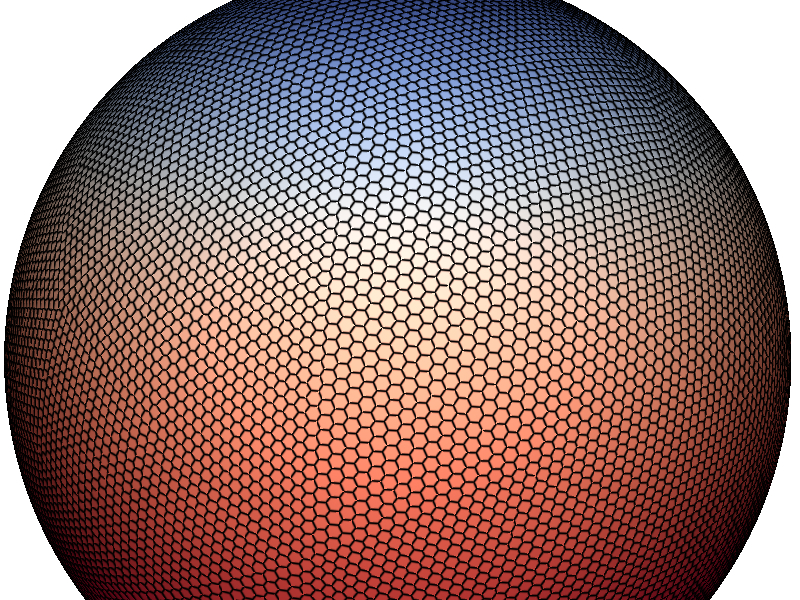}
        \caption{0 days.}
    \end{subfigure}
    \begin{subfigure}[t]{0.22\textwidth}
        \includegraphics[width=\textwidth]{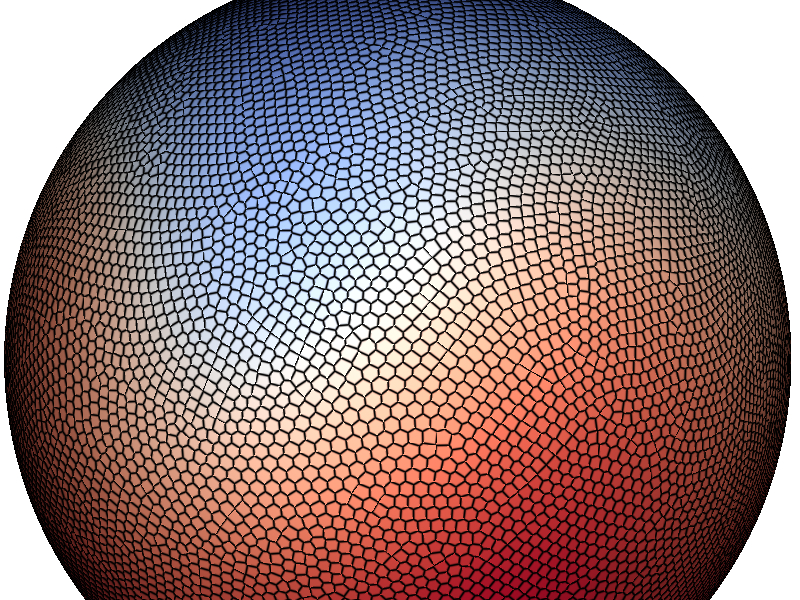}\\
        \includegraphics[width=\textwidth]{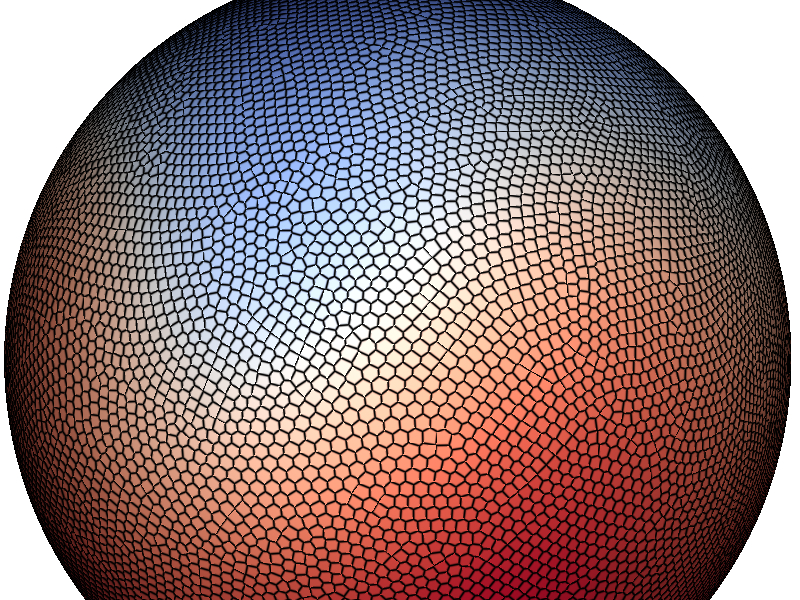}
        \caption{5 days.}
    \end{subfigure}
    \begin{subfigure}[t]{0.225\textwidth}
        \includegraphics[width=\textwidth]{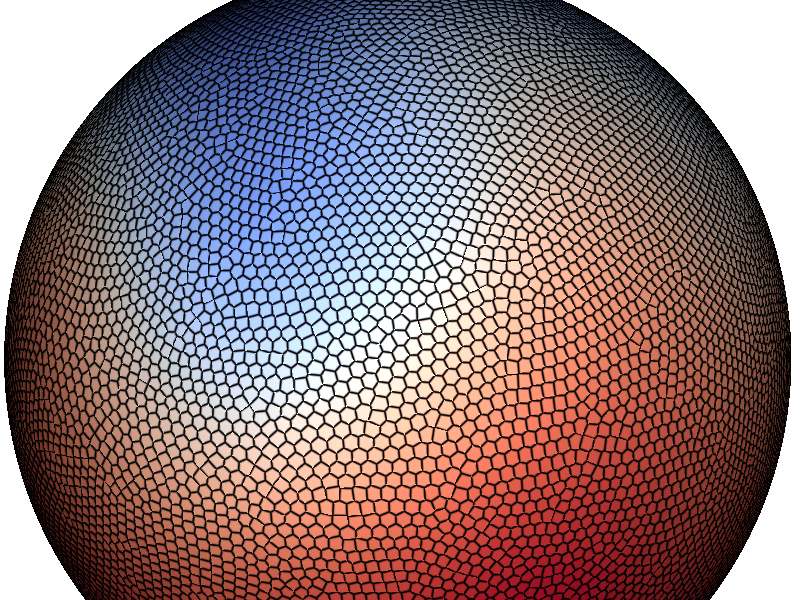}\\
        \includegraphics[width=\textwidth]{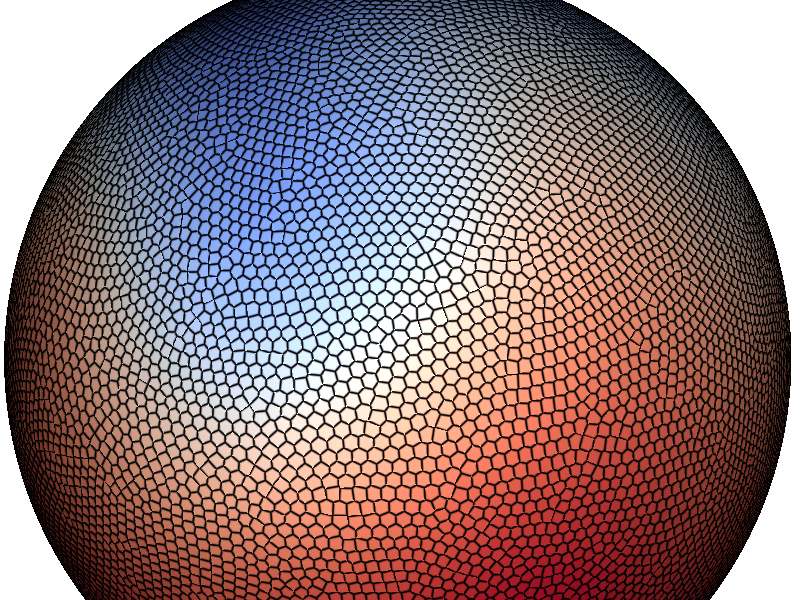}
        \caption{10 days.}
    \end{subfigure}
    \begin{subfigure}[t]{0.225\textwidth}
        \includegraphics[width=\textwidth]{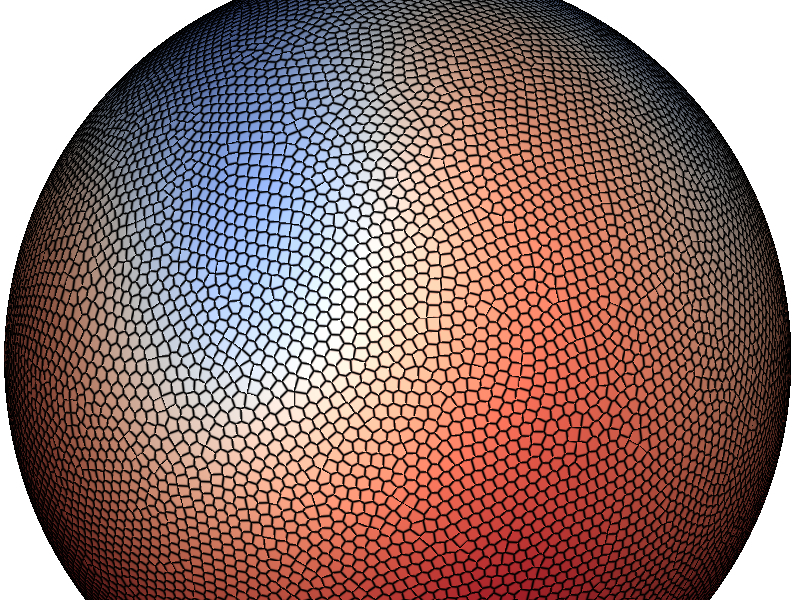}\\
        \includegraphics[width=\textwidth]{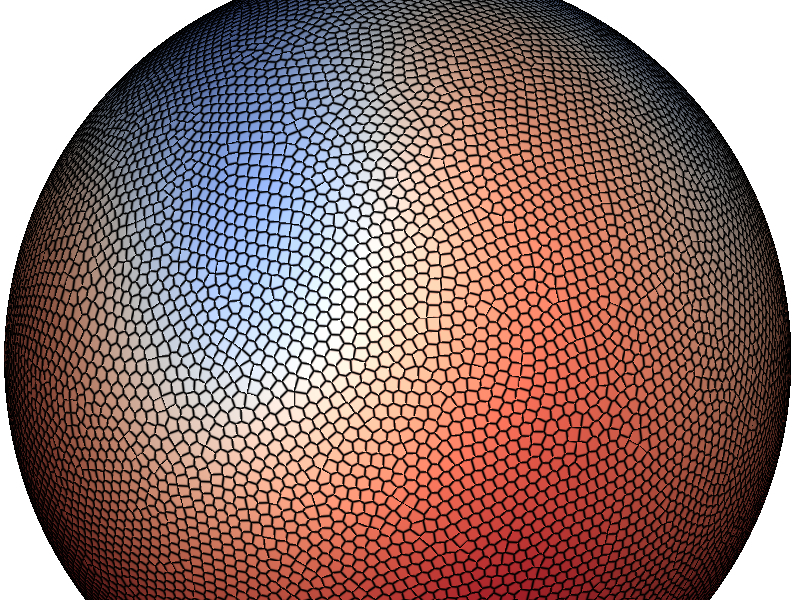}
        \caption{15 days.}
    \end{subfigure}
    \caption{$3d$ view of the height field (top row) for Williamson Test Case 5 at 0, 5, 10 and 15 days computed with 40,962 particles and a time step of 2 minutes.
        The bottom row shows the weights of the power diagram at that time step which is zero at the beginning of the simulation but then increases upstream of the mountain as the simulation runs.
    }
    \label{fig:w5-hw-3d}
\end{figure}
The accuracy of the solver is assessed by using a TRiSK-based finite volume solver called \texttt{swe-python}~\cite{Engwirda_swe_python} to generate a reference solution on a centroidal Voronoi tessellation with 655,362 cells and 1,310,720 vertices.
To measure the difference in the solution at a given particle (for a chosen time), the solution at the closest cell center (in the reference solution) is used as the reference solution for that particle.
The normalized difference $\Delta I_h = \sqrt{\sum_{i=1}^n (h_{i} - h_{i,*})^2 / \sum_{i=1}^n h_{i,*}^2}$ between our particle heights $h_i$ and the height $h_i^*$ of the closest cell in the reference solution for 10,242, 40,962, and 163,842 particles is shown in Fig.~\ref{fig:w5-comparison}.
It should be noted that our solution does not agree with the reference solution at day 0, perhaps due to the cell-wise smoothing procedure used by the reference solver when applying the initial condition.
Furthermore, momentum and energy conservation are comparable to the trends reported in Capecelatro's work~\cite{Capecelatro_2018}.
\begin{figure}
    \centering
    \begin{subfigure}[t]{0.49\textwidth}
        \includegraphics[width=\textwidth]{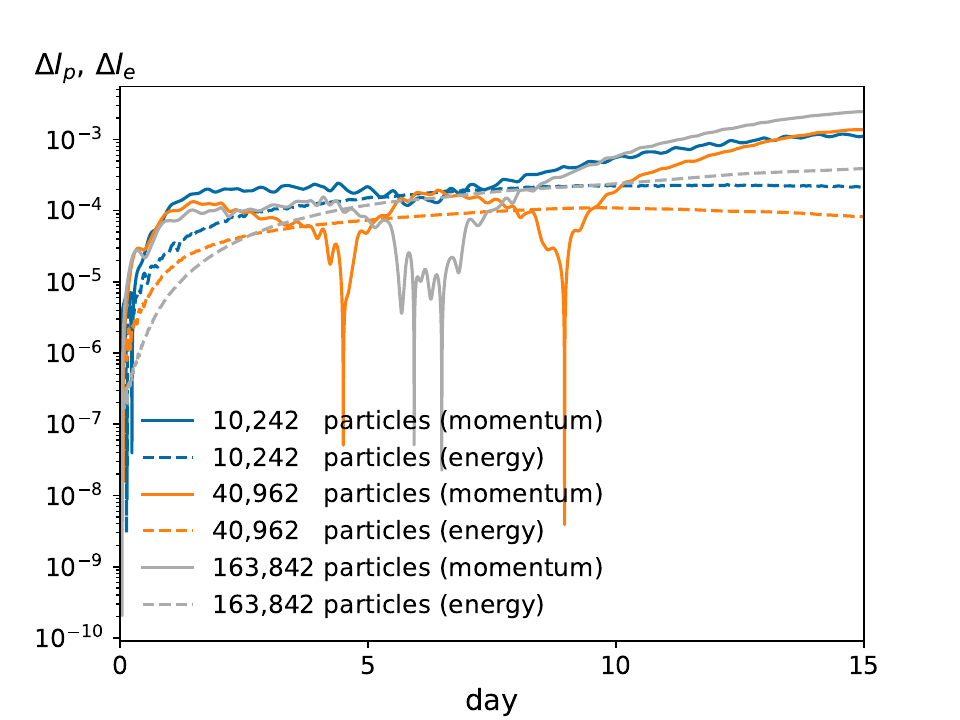}
        \caption{Conservation of momentum (dashed) and energy (solid) over 15 days of simulation.}
        \label{fig:w5-conservation}
    \end{subfigure}
    \begin{subfigure}[t]{0.49\textwidth}
        \includegraphics[width=\textwidth]{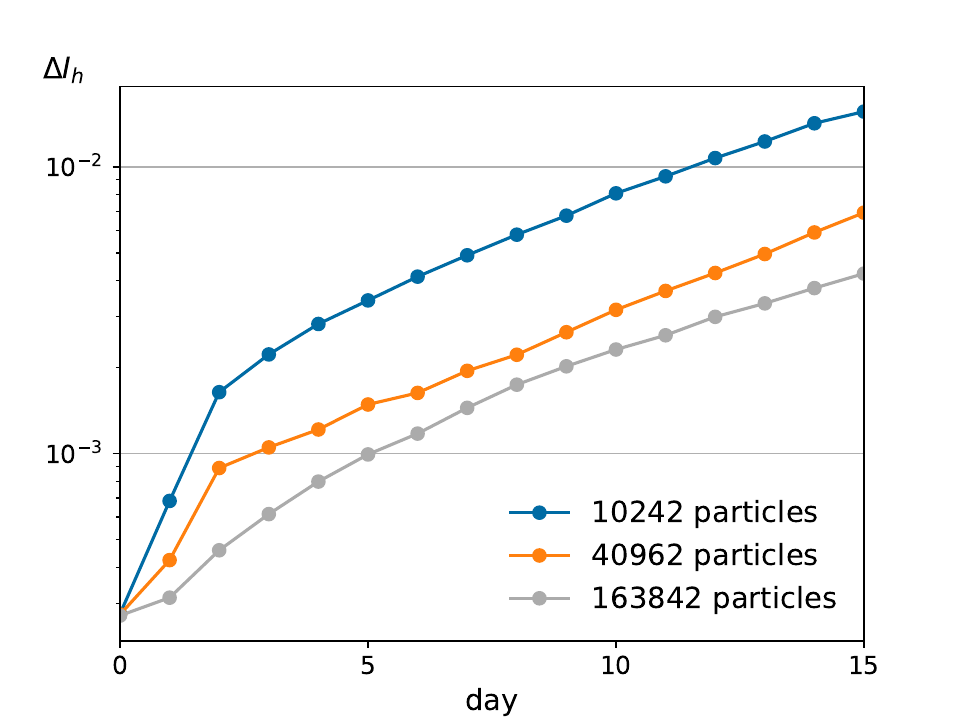}
        \caption{Difference between our computed height field with that obtained from \texttt{swe-python} over 15 days of simulation.}
        \label{fig:w5-comparison}
    \end{subfigure}
    \caption{Conservation of momentum (dashed) and energy (solid) over 15 days of simulation for Williamson Test Case 5.
        This test case was also compared with \texttt{swe-python}~\cite{Engwirda_swe_python}.
    }
\end{figure}
\subsection{Williamson Test Case 6: Rossby-Haurwitz wave}
The final test case consists of simulating the wavenumber 4 Rossby-Haurwitz wave.
Although this is not an analytic solution to the governing equations, it is commonly used to benchmark shallow water equation solvers.
Please refer to the complete description of the initial surface velocity and height expressions in the Williamson et al. paper~\cite{Williamson_1992}.
As in the previous cases, time steps of $\Delta t = 2\,\mathrm{minutes},\ 1\,\mathrm{minute}$, and $30\,\mathrm{seconds}$ for 10,242, 40,962, and 163,842 particles, respectively, were used in the simulations.
The height field is shown at days 0, 5, 10 and 15 in Fig.~\ref{fig:w6-h-solution}.
While the wave roughly retains its shape for all particle numbers and time steps, the final solution does appear somewhat smoothed relative to the initial condition.
Among the three cases studied in this work, conservation of energy and momentum (Fig.~\ref{fig:w6-conservation}) was the poorest for the Rossby-Haurwitz wave, with $\Delta I_p$ and $\Delta I_e$ lying between $10^{-2}$ and $10^{-3}$.
\begin{figure}
    \begin{subfigure}{0.49\textwidth}
        \includegraphics[width=\textwidth]{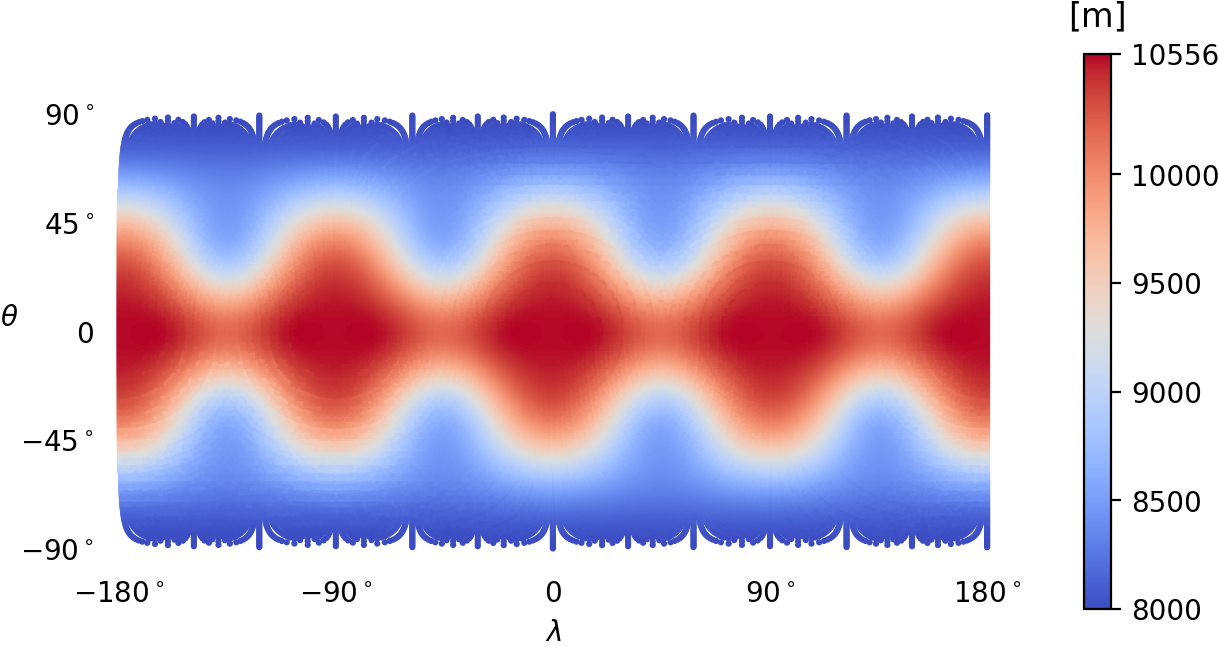}
        \caption{Initial height.}
    \end{subfigure}
    \begin{subfigure}{0.49\textwidth}
        \includegraphics[width=\textwidth]{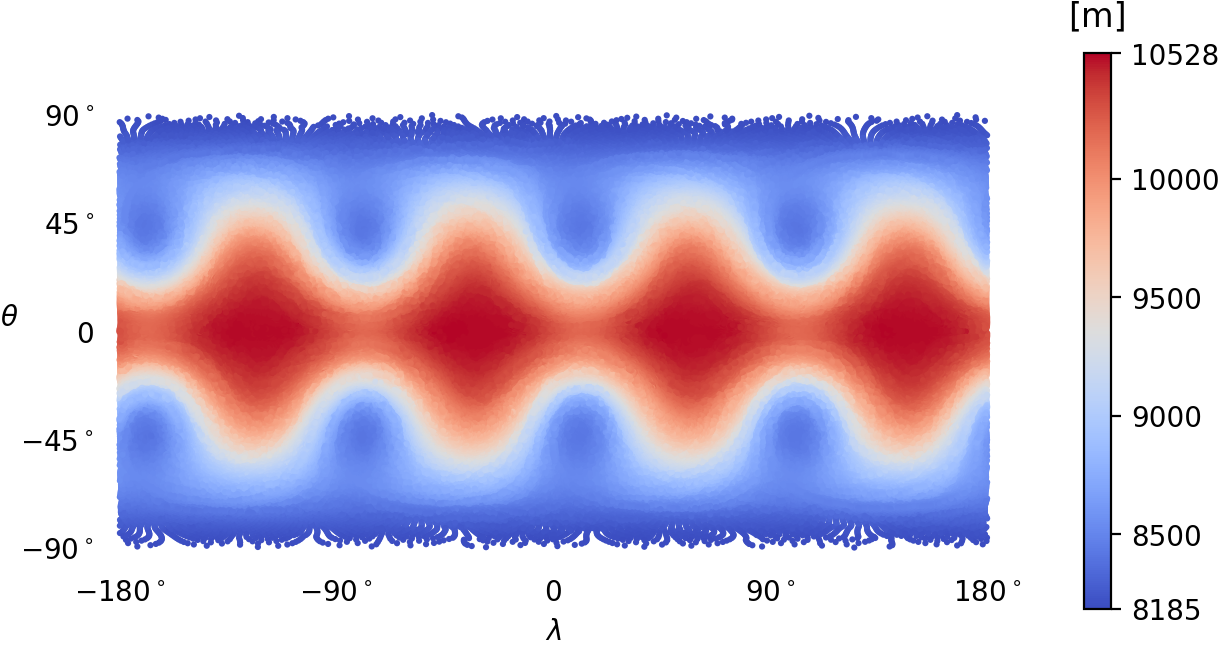}
        \caption{Height at 5 days.}
    \end{subfigure}
    \begin{subfigure}{0.49\textwidth}
        \includegraphics[width=\textwidth]{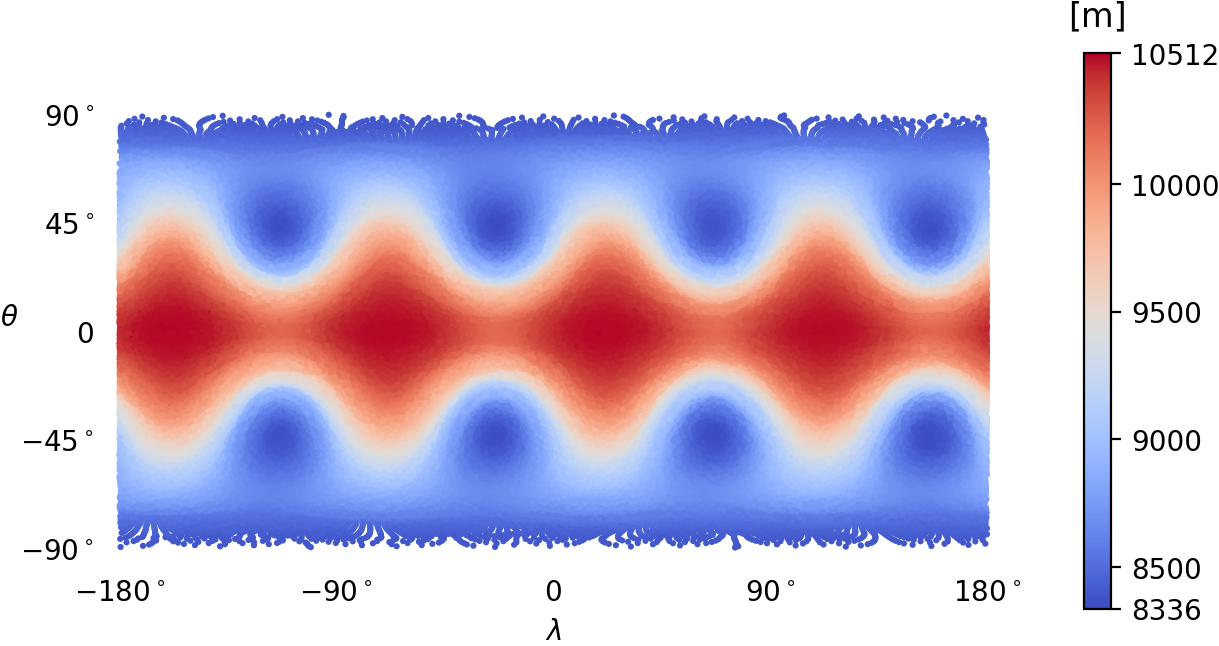}
        \caption{Height at 10 days.}
    \end{subfigure}
    \begin{subfigure}{0.49\textwidth}
        \includegraphics[width=\textwidth]{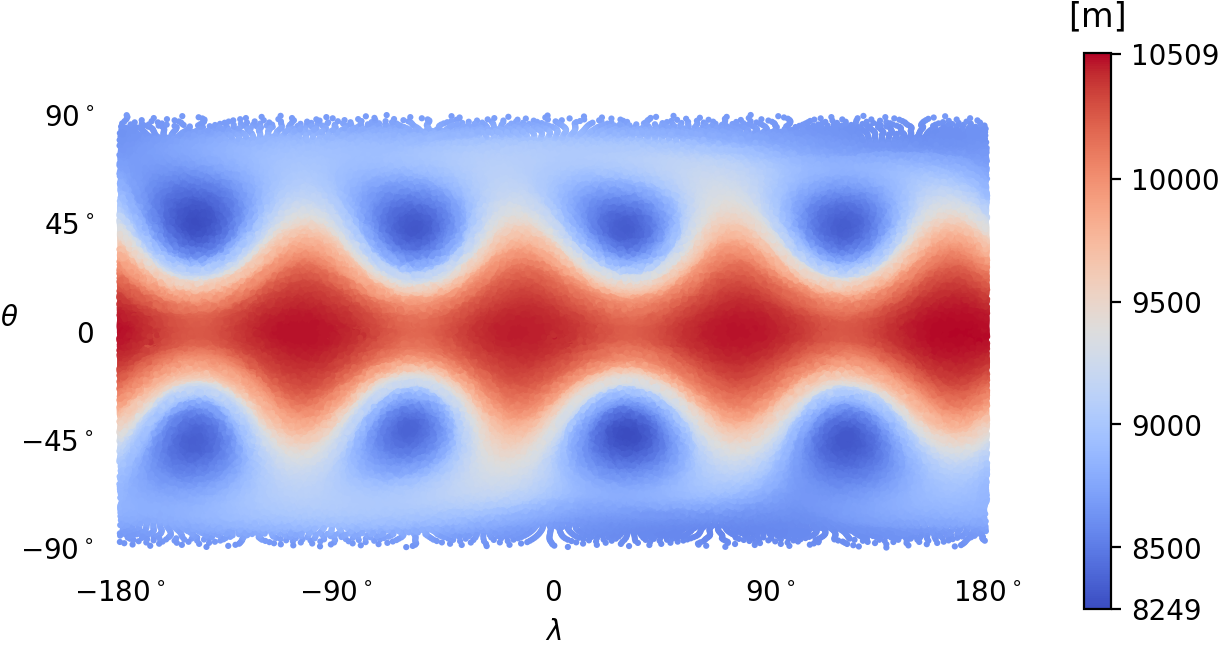}
        \caption{Height at 15 days.}
    \end{subfigure}
    \caption{Height field for Rossby-Haurwitz Wave test at 0, 5, 10 and 15 days computed with 163, 842 particles and a time step of 30 seconds.}
    \label{fig:w6-h-solution}
\end{figure}

Finally, we use this case to comment on the performance of all the components of the solver.
Earlier sections evaluated the performance of the Voronoi diagram construction, which we now situate within the time for a complete time step.
Fig.~\ref{fig:w6-timing} shows the timing of a complete time step with 163,842 particles (for all time steps over the 15-day simulation) along with the percentage consumed by the Voronoi diagram calculations and the linear solver time, which is needed to solve both Eq.~\ref{eq:newton-update} for the Newton update and Eq.~\ref{eq:semi-implicit-height} for the semi-implicit height update.
Each time step took about 2.4 seconds on an AMD EPYC 7513 with 32 cores.
Roughly 12\% of that time was used to compute all the Voronoi diagrams (on average, 3 per time step) needed to solve the semi-discrete optimal transport problem.
The linear solver (to solve both systems of equations) was the most costly component, consuming on average 70\% of each time step.
The remaining time was used to update the particle heights and velocities (which is not parallelized for now) and to calculate the differential operators.
\begin{figure}
    \centering
    \begin{subfigure}[t]{0.49\textwidth}
        \includegraphics[width=\textwidth]{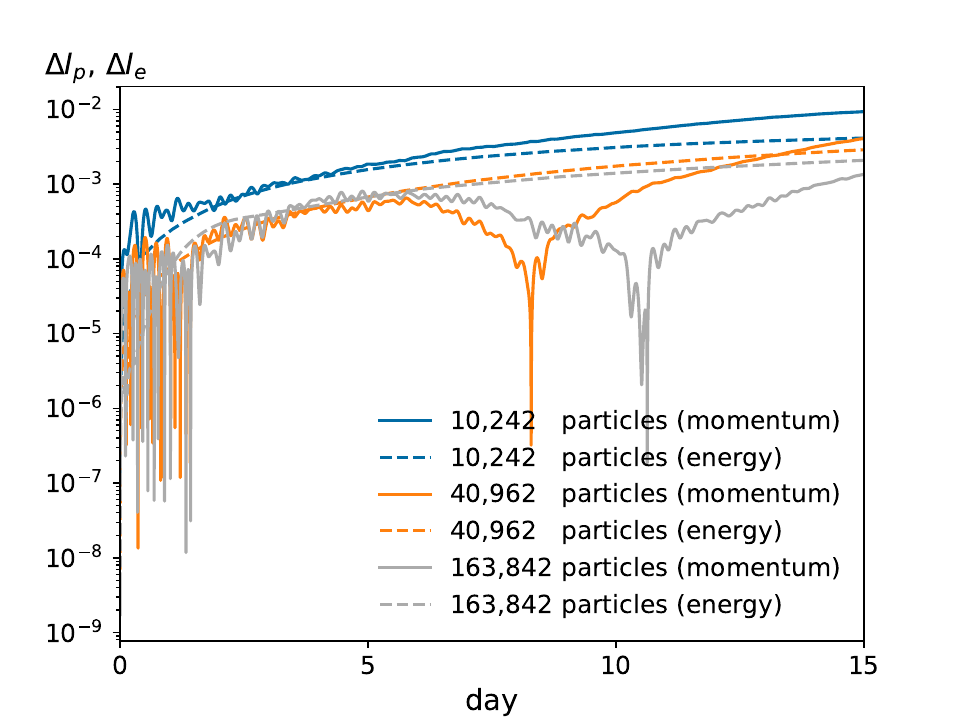}
        \caption{Conservation of momentum (dashed) and energy (solid).}
        \label{fig:w6-conservation}
    \end{subfigure}
    \begin{subfigure}[t]{0.49\textwidth}
        \includegraphics[width=\textwidth]{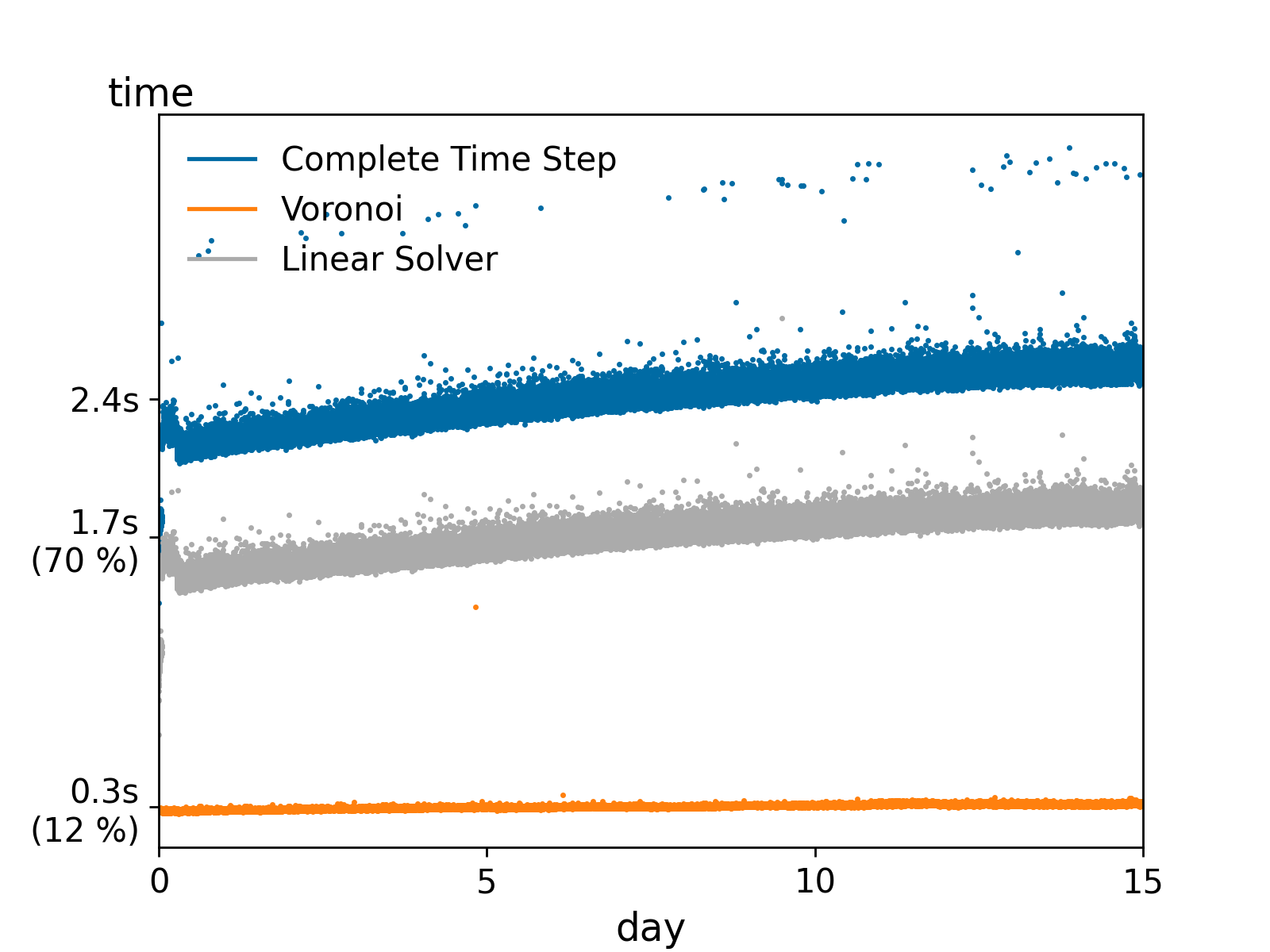}
        \caption{Breakdown of the time each component of a time step took using an AMD EPYC 7513 (32 cores).}
        \label{fig:w6-timing}
    \end{subfigure}
    \caption{Conservation properties over 15 days of simulation for the Rossby-Haurwitz Wave test.
        This case was also used to analyze the timing breakdown of each time step.
    }
\end{figure}
\section{Conclusions and Future Work}
This paper makes progress towards determining whether Lagrangian simulations of geophysical fluids can outperform Eulerian simulations.
We introduced a new Lagrangian method for solving the shallow water equations on the sphere which uses weighted Voronoi cells to represent the particles.

We started by extending a halfspace clipping algorithm to compute spherical Voronoi cells in which Voronoi vertices are computed using the intersection between a sphere and the line of intersection between two Voronoi bisectors.
The spherical Voronoi diagram calculation is faster than existing methods and is able to calculate a Voronoi diagram of 100 million sites in under 2 minutes.
These Voronoi diagrams then lent to the construction of differential operators that were used to discretize the inviscid shallow water equations.
Integration in time was performed with a semi-implicit scheme, which was robust for the cases studied in this work.
Our method also does not require artificial viscosity.
Overall, the method exhibited the same level of momentum and energy conservation as the most recent Lagrangian method for simulating the shallow water equations~\cite{Capecelatro_2018}, but our results appear more consistent with those of Eulerian solvers. 
For the second case of Williamson et al.~\cite{Williamson_1992}, the error in the height and velocity converged with a rate of 1 as the particle sizes decreased.
Furthermore, the method was in agreement with an Eulerian shallow water equation solver for the case of zonal flow over an isolated mountain.

The results are promising, and there are several avenues for future work.
For the cases studied here, the method was robust, and requires no parameters other than number of particles and a desired time step.
Accuracy may be improved with higher-order approximations of the differential operators in the discretization.
It may also be worthwhile to investigate whether the moving finite volume discretization of Springel~\cite{Springel_2010} would improve momentum and energy conservation, in addition to implementing a vorticity-divergence formulation.
Performance may be improved by using a GPU-based approach to compute Voronoi cells~\cite{Ray_2018, Basselin_2021, Templeton_2022} as well as libraries such as \texttt{AMGCL}~\cite{Demidov_2020} to solve the linear systems.
Finally, we hope to extend this method to simulate the oceans by keeping large, fixed particles in land areas with moving particles in the oceans as in Fig.~\ref{fig:voronoi-oceans}.

Further development of these ideas are welcome by collaborating at\\ \texttt{https://github.com/philipclaude/vortex}.
\begin{figure}
    \centering
    \includegraphics[width=0.99\textwidth]{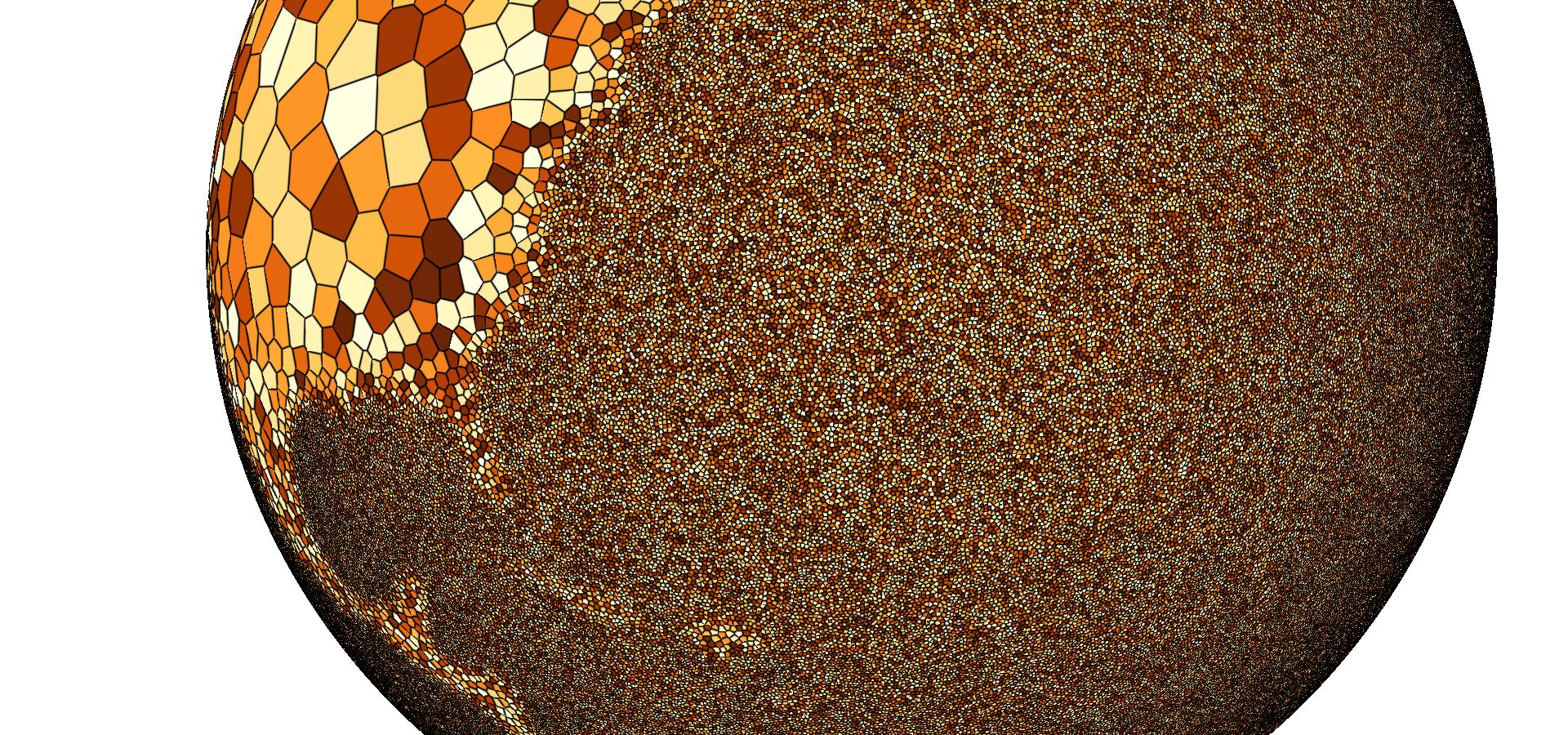}
    \caption{Preliminary work at setting up a Voronoi diagram with higher resolution in the oceans.
        This Voronoi diagram contains 1 million sites and contains less than 2\% of the sites on land~\cite{Pouler_2024}.
    }
    \label{fig:voronoi-oceans}
\end{figure}

\section*{Acknowledgments}
This material is based upon work supported by the National Science Foundation under Grant No. 1827373.

%
%
%
\appendix
\section{Lifting sites to compute power cells}
\label{app:lifting}
Aurenhammer~\cite{Aurenhammer_1987} showed that power cells can be computed from the Voronoi cells of sites lifted to $\mathbb{R}^4$, and L\'evy~\cite{Levy_2015} provided an explicit expression for the lifted coordinate. This connection is elaborated upon below and adapted to the case of a sphere. As mentioned in Remark~\ref{rem:voronoi-equivalence}, either approach works, and the following results are simply meant to provide the theoretical basis for this connection in case the reader wishes to implement either approach. One particular advantage of the lifting approach is that standard nearest neighbor approaches (e.g. using a kd-tree) can be used.
\begin{proposition}
    \label{thm:lifting}
    Let $L_0\,\colon \mathbb{R}^3 \to \mathbb{R}^4$ be an embedding in which the lifted (fourth) coordinate is zero:
    \begin{align*}
        L_0\left(\vec{r}\right) = \mathbf{M}\vec{r},\quad \forall\,\vec{r} \in \mathbb{R}^3,\,\,\mathrm{with}\,\, \mathbf{M} =  \left[\begin{array}{ccc}
        1 & 0 & 0 \\
        0 & 1 & 0 \\
        0 & 0 & 1 \\
        0 & 0 & 0
        \end{array}\right].
    \end{align*}
    Note that $L_0^{-1}\,\colon\mathbb{R}^4 \to \mathbb{R}^3$ projects a point in $\mathbb{R}^4$ back to $\mathbb{R}^3$ by discarding the fourth coordinate, and $\mathbf{M}^T\mathbf{M}\vec{r} = \vec{r}$.
    With $\vec{r}_{\ell} = \mathbf{M}\vec{r}$, define
    $\widehat{\mathbb{S}}^2 = L_0(\mathbb{S}^2) = \{ \vec{r}_\ell\,\colon\,\lVert \vec{r}_\ell\rVert = a\}$ as the sphere $\mathbb{S}^2$ with radius $a$ lifted to $\mathbb{R}^4$. A spherical power diagram of the sites $\{\vec{p}_i = (x_i, y_i, z_i) \in \mathbb{R}^3\}$ with weights $\{w_i \in \mathbb{R}\}$ is the spherical Voronoi diagram of the sites $\left\{ \vec{q}_i = (x_i, y_i, z_i, \sqrt{w_{\max} - w_i}) \right\}$ where $w_{\max} \ge \max\{w_i\}$, projected back to $\mathbb{R}^3$ by $L_0^{-1}$. In other words,
    \begin{align*}
        \mathbb{P}_i = L_0^{-1}(\widehat{\mathbb{V}}_i),\,\,\mathrm{where}\,\,\widehat{\mathbb{V}}_i = \{ \vec{r}_\ell \in \widehat{\mathbb{S}}^2\,\colon\,\lVert\vec{q}_i - \vec{r}_\ell\rVert < \lVert \vec{q}_j - \vec{r}_\ell\rVert,\ \forall j \neq i\}.
    \end{align*}%
\end{proposition}%
\begin{proof}
    Note that $L_0^{-1}(\widehat{\mathbb{S}}^2) = \mathbb{S}^2$. We need to show that the expression above for $\mathbb{P}_i$ matches the original definition of a spherical power cell (Eq.~\ref{eq:power-cell}). By substituting $\vec{r}_\ell = \mathbf{M}\vec{r}$, as well as the definition of $\vec{q}_i = (x_i, y_i, z_i, \sqrt{w_{\max} - w_i})$, observe that $\lVert \vec{q}_i - \mathbf{M}\vec{r}\rVert^2 = \lVert\vec{p}_i - \vec{r}\rVert^2 + w_{\max} - w_i$. Similarly, $\lVert \vec{q}_j - \mathbf{M}\vec{r}\rVert^2 = \lVert\vec{p}_j - \vec{r}\rVert^2 + w_{\max} - w_j$. Therefore,
    \begin{align*}
        \widehat{\mathbb{V}}_i = \{ \mathbf{M}\vec{r} \in \widehat{\mathbb{S}}^2\,\colon\,\lVert\vec{p}_i - \vec{r}\rVert^2 + w_{\max} - w_i  < \lVert\vec{p}_j - \vec{r}\rVert^2 + w_{\max} - w_j,\ \forall j \neq i\}.
    \end{align*}
    Applying $L_0^{-1}$ to the domain (i.e. premultiplying by $\mathbf{M}^T$) yields:
    \begin{align*}
        L_0^{-1}(\widehat{\mathbb{V}}_i) = \{ \vec{r} \in \mathbb{S}^2\,\colon\,\lVert\vec{p}_i - \vec{r}\rVert^2 - w_i  < \lVert\vec{p}_j - \vec{r}\rVert^2 - w_j,\ \forall j \neq i\} = \mathbb{P}_i,
    \end{align*}%
\end{proof}%
We can further connect the halfspace-based definition of Eq.~\ref{eq:power-cell-halfspaces} to a halfspace-based definition with the lifted sites. Specifically, the halfspace-based definition of Voronoi cells in $\mathbb{R}^4$ is:
\begin{align}
    \label{eq:voronoi-cell-halfspaces-lifted}
    \widehat{\mathbb{V}}_i = \bigcap\limits_{j \neq i}\ \widehat{\mathcal{H}}^+_{ij}, \quad\widehat{\mathcal{H}}_{ij}^+ = \{ \vec{r}_\ell \in \widehat{\mathbb{S}}^2\,\colon\,(\vec{r}_\ell - \vec{b}_{ij,\ell})\cdot \vec{n}_{ij,\ell} > 0 \}.
\end{align}
where $\vec{b}_{ij,\ell} = \frac{1}{2}(\vec{q}_i + \vec{q}_j) \in \mathbb{R}^4$ is a point on the bisector and $\vec{n}_{ij,\ell} = \vec{q}_i - \vec{q}_j \in \mathbb{R}^4$ is the normal to the bisector (in the lifted space), which points into $\widehat{\mathbb{V}}_i$.
\begin{proposition}
    \label{thm:halfspaces}
    The halfspace-based definition of a Voronoi cell $\widehat{\mathbb{V}}_i$ in Eq.~\ref{eq:voronoi-cell-halfspaces-lifted} is equivalent to the halfspace-based definition of a power cell $\mathbb{P}_i$ in Eq.~\ref{eq:power-cell-halfspaces}.
\end{proposition}
\begin{proof}
    The main task is to show the equivalence between the two predicates in the sets defining the halfspaces $\mathcal{H}_{ij}^+$ (Eq.~\ref{eq:power-cell-halfspaces}) and $\widehat{\mathcal{H}}_{ij}^+$ (Eq.~\ref{eq:voronoi-cell-halfspaces-lifted}). First, note that
    \begin{align*}
        \vec{b}_{ij}\cdot\vec{n}_{ij} = \left(\frac{1}{2}(\vec{p}_i + \vec{p}_j) + \frac{1}{2}(w_i - w_j) \frac{(\vec{p}_j - \vec{p}_i)}{\lVert\vec{p}_j - \vec{p}_i\rVert^2}\right)\cdot(\vec{p}_i - \vec{p}_j) =
        \frac{1}{2}(\lVert\vec{p}_i\rVert^2 + \lVert\vec{p}_j\rVert^2 - w_i + w_j).
    \end{align*}
    Starting with the predicate of Eq.~\ref{eq:voronoi-cell-halfspaces-lifted} and expanding $\vec{n}_{ij,\ell}$ and $\vec{b}_{ij,\ell}$:
    \begin{align*}
        (\vec{r}_\ell - \vec{b}_{ij,\ell})\cdot \vec{n}_{ij,\ell} &= \vec{r}_\ell \cdot (\vec{q}_i - \vec{q}_j) - \frac{1}{2}(\vec{q}_i + \vec{q}_j) \cdot (\vec{q}_i - \vec{q}_j) & \\
        & = \vec{r}\cdot (\vec{p_i} - \vec{p}_j) - \frac{1}{2}(\vec{q}_i + \vec{q}_j) \cdot (\vec{q}_i - \vec{q}_j)\ \mbox{(since the fourth component of $\vec{r}_\ell$ is 0),} \\
        &= \vec{r}\cdot \vec{n} - \frac{1}{2}(\vec{q}_i + \vec{q}_j) \cdot (\vec{q}_i - \vec{q}_j) \ \mbox{(substituting $\vec{n} = \vec{p}_i - \vec{p}_j$),} \\
        &= \vec{r}\cdot \vec{n} - \frac{1}{2}(\lVert\vec{q}_i\rVert^2 - \lVert\vec{q}_j\rVert^2), & \\
        &= \vec{r}\cdot \vec{n} - \frac{1}{2}(\lVert\vec{p}_i\rVert^2 + w_{\max} - w_i - \lVert\vec{p}_j\rVert^2 - w_{\max} + w_j), & \\
        &= \vec{r}\cdot \vec{n} - \frac{1}{2}(\lVert\vec{p}_i\rVert^2 - \lVert\vec{p}_j\rVert^2 - w_i + w_j),
    \end{align*}
    which is the same as $\vec{r}\cdot\vec{n} - \vec{b}_{ij}\cdot\vec{n}_{ij} = (\vec{r} - \vec{b}_{ij})\cdot\vec{n}_{ij}$ used to describe $\mathcal{H}_{ij}^+$.
\end{proof}

\bibliography{references}
\end{document}